\documentclass[pra,twocolumn,floatfix,notitlepage,superscriptaddress,longbibliography]{revtex4-1}                                                                              
\usepackage{mathtools}
\usepackage{scalerel,stackengine}
\stackMath
\newcommand\reallywidehat[1]{%
\savestack{\tmpbox}{\stretchto{%
  \scaleto{%
    \scalerel*[\widthof{\ensuremath{#1}}]{\kern-.6pt\bigwedge\kern-.6pt}%
    {\rule[-\textheight/2]{1ex}{\textheight}}%WIDTH-LIMITED BIG WEDGE
  }{\textheight}% 
}{0.5ex}}%
\stackon[1pt]{#1}{\tmpbox}%
}
\usepackage{amsmath}
\usepackage{ytableau}
\usepackage{amssymb,amsmath,amstext,amsthm}
\usepackage{graphicx}
\usepackage{epstopdf}
\usepackage{color}
\usepackage{dsfont}
\usepackage{bm}
\usepackage{appendix}
\usepackage[T1]{fontenc}
\usepackage{bbm}
\usepackage{latexsym}
\usepackage{bbm}
\usepackage{float}
\usepackage{verbatim}
\usepackage{lipsum}
\usepackage{braket}
\usepackage{ulem}
\usepackage{mathtools}
\usepackage{mathrsfs}
\usepackage{dcolumn}
\usepackage{xcolor}   
\usepackage{bm}        
\usepackage{amssymb}   % for math
\usepackage{physics}
\usepackage{tikz}
\usepackage{gensymb}
\usepackage{xcolor}
\usepackage{float}
\usepackage{dsfont}
\usepackage[caption = false]{subfig}
\usepackage{mathtools}
 \usepackage{amsthm}
 \usepackage[english]{babel}
 \usepackage[bookmarks,bookmarksopen,bookmarksdepth=2]{hyperref}
\hypersetup{
    colorlinks=true,
    citecolor=blue,
    linkcolor=blue,
    filecolor=magenta,
    urlcolor=blue,
}
 \usepackage{quantikz}
\usepackage[capitalize]{cleveref}
\newcommand{\SU}{\mathrm{SU}}

\newtheorem{lemma}{Lemma}

\def\frac#1#2{{\begingroup #1\endgroup\over #2}}

\renewcommand{\leq}{\leqslant}

\def \addPandAUNM {Department of Physics and Astronomy, University of New Mexico, Albuquerque, NM, USA}
\def \addCQuIC {Center for Quantum Information and Control, University
  of New Mexico, Albuquerque, NM, USA}
\def \addLANL {Theoretical Division, Los Alamos National Laboratory, Los Alamos, NM 87545, USA}
\def \addJPMorgan{Global Technology Applied Research, JPMorganChase, New York, NY 10017, USA}

\begin{document}

\title{Quantum error correction-inspired multiparameter quantum metrology}

\author{Sivaprasad  Omanakuttan}
\email{sivaprasad.thattupurackalomanakuttan@jpmchase.com}
\affiliation{\addJPMorgan}
\affiliation{\addCQuIC} \affiliation{\addPandAUNM} 
\author{Jonathan A. Gross }
\email[]{jarthurgross@google.com}
\affiliation{Google Quantum AI, Venice, CA 90291, USA
}
\author{T.J.\,Volkoff}
\email{volkoff@lanl.gov}
\affiliation{\addLANL}

\begin{abstract}
We present a novel strategy for obtaining optimal probe states and measurement schemes in a class of noiseless multiparameter estimation problems with symmetry among the generators.
The key to the framework is the introduction of a set of quantum metrology conditions, analogous to the quantum error correction conditions of Knill and Laflamme, which are utilized to identify probe states that saturate the multiparameter quantum Cram\'{e}r-Rao bound.
Similar to finding two-dimensional irreps for encoding a logical qubit in error correction, we identify trivial irreps of finite groups that guarantee the satisfaction of the quantum metrology conditions.
To demonstrate our framework, we analyze the SU(2) estimation with symmetric states in which three parameters define a global rotation of an ensemble of $N$ qubits.
For even $N$, we find that tetrahedral symmetry and, with fine-tuning, $S_{3}$ symmetry, are minimal symmetry groups providing optimal probe states for SU(2) estimation, but that the quantum metrology conditions can also be satisfied in an entanglement-assisted setting by using a maximally entangled state of two spin-$N/2$ representations for any $N$.
By extending the multiparameter method of moments to non-commuting observables, we use the quantum metrology conditions to construct a measurement scheme that saturates the multiparameter quantum Cram\'{e}r-Rao bound for small rotation angles.
\end{abstract}
\maketitle

\section{Introduction\label{sec:intro}}
Estimation of dynamical parameters of quantum systems, i.e., quantum metrology \cite{Braunstein_caves_metrology,paris2009quantum,Tóth_2014,Liu_2020}, holds significant importance for a variety of quantum technology applications, including frequency spectroscopy using atomic clocks \cite{RevModPhys.87.637}, gravitational wave detection \cite{PhysRevD.65.022002}, and electric and magnetic field estimation in NV centers \cite{bian2021nanoscale,doi:10.1126/science.aam7009}.
The general goal of quantum metrology is to estimate unknown parameters with the highest possible precision.
Typically, a one-shot quantum estimation protocol involves four key steps \cite{Liu_2020}: (1) preparing the probe state, (2) parametrization wherein the parameters to be estimated are encoded in the probe state, (3) measuring the parameterized probe state, and (4) classical estimation. 
Because the final step is covered by the classical statistical theory, quantum metrology primarily focuses on optimizing the first three steps in order to minimize the local estimation error.

Lower bounds on the local estimation error have been extensively studied, resulting in quantum generalizations of the Cram\'{e}r-Rao bound \cite{helstrom1969quantum,holevo}.
Within the framework of these quantum Cram\'{e}r-Rao bounds, quantum generalizations of the Fisher information matrix provide quantities that define the precision limits for both single-parameter and multiparameter estimation.
For example, the SLD quantum Fisher information matrix provides the achievable quantum Cram\'{e}r-Rao bound for estimating any single parameter in isolation~ \cite{Braunstein_caves_metrology}.
In the case of estimating the angle of an SU(2) rotation about an axis on a system of $N$-qubits, preparing an equal-weight superposition of the highest and lowest eigenvectors of the rotation generator (i.e., a GHZ state) and making a parity measurement saturates this bound~ \cite{PhysRevA.54.R4649}.
% For example, when the symmetric logarithmic derivative (SLD) operators are used as a basis for the tangent space of the quantum state manifold, the SLD quantum Fisher information matrix provides the appropriate quantum Cram\'{e}r-Rao bound.
% For single parameter estimation, this lower bound is achievable \cite{Braunstein_caves_metrology}, as exemplified for an estimator of a single-parameter SU(2) rotation of a system of $N$-qubits prepared in an equal-amplitude superposition of the highest and lowest eigenvectors of the rotation generator (i.e., a GHZ state). In this well-known example, the SLD quantum Cram\'{e}r-Rao bound is $1/N$ for the estimation error, providing the notion of Heisenberg limit. 
% This optimal precision can be achieved by, e.g., a parity measurement \cite{PhysRevA.54.R4649}.

By contrast, many important scenarios, including microscopy, optics, electromagnetic studies, gravitational field imaging, and spectroscopy, involve the estimation of multiple parameters associated with terms in the Hamiltonian that do not commute. 
Thus there is growing interest in multiparameter problems, which can make use of many degrees of freedom of a quantum probe system. 
For example, given local qubit Hamiltonians $H_{k}$ with tunable coupling parameters $\theta_{k}$, algorithms exist which estimate the $\theta_{k}$ in a parametrized state $e^{-i\sum_{k}\theta_{k}H_{k}}\rho e^{i\sum_{k}\theta_{k}H_{k}}$ to an error $O(1/N)$  with high probability \cite{PhysRevLett.130.200403}.
It is therefore of increasing practical importance to identify and generate optimal probe states that allow multiparameter estimation algorithms to be carried out with the highest sensitivity and efficiency.

In this setting of multiparameter estimation, formulating achievable lower bounds for classes of estimation problems remains a 
challenging task.
For instance, the multiparameter SLD quantum Cram\'{e}r-Rao bound is only achievable in very specific models. Alternative descriptions of the tangent space can lead to achievable quantum Cram\'{e}r-Rao bounds, as in the problem of the Gaussian state model \cite{holevo} in which the quantum Cram\'{e}r-Rao bound based on right logarithmic derivative operators is achievable and strictly greater than the SLD quantum Cram\'{e}r-Rao bound.
More generally, in multiparameter estimation one must resort to tighter, non-explicit quantum Cram\'{e}r-Rao bounds like the Nagaoka-Hayashi bound \cite{nagaoka,hayashi}, or Holevo Cram\'{e}r-Rao bound, or seek a well-motivated measurement scheme that can be proven to locally satisfy the SLD quantum Cram\'{e}r-Rao bound.

Quantum error correction is an important tool to protect information from noise. 
In quantum error correction, one encodes information in a small logical Hilbert space, which is a subspace of an often much larger physical Hilbert space that is protected from certain physical errors in the sense that an appropriate correction channel exists for these errors \cite{PhysRevA.55.900}. 
Quantum error correction techniques have been used to explore the possibility of achieving the Heisenberg limit for parameter estimation in the presence of noise \cite{PhysRevLett.112.150802,Zhou2018,PhysRevA.100.022312,dur2014,PhysRevLett.112.150801,PhysRevLett.112.150802,PhysRevLett.122.040502,Gorecki2020optimalprobeserror}. In these works, the errors have the straightforward interpretation as occurring due to non-unitary quantum dynamics.   
In this work, we connect the concept of quantum error correction to multiparameter quantum metrology in a noiseless setting, using the well-studied criteria for quantum error correction to motivate criteria for multiparameter quantum metrology that indicate the optimality of a probe state quantified by the mean squared error of the parameter estimates.
% its SLD quantum Fisher information matrix. 

In particular, we derive a set of conditions that identify optimal states for estimation of parameters coupled to generators that do not commute but are related by a discrete commutative symmetry.
These conditions mirror the Knill-Laflamme criteria \cite{PhysRevA.55.900} for pure quantum error correcting codes, and further motivate measurement schemes that locally saturate the SLD quantum Cram\'{e}r-Rao bound.
Notably, in our setting, optimal multiparameter quantum metrology can be viewed as a form of quantum error correction where the objective is to encode a state that is robust to errors given by the generators of the quantum metrology problem while maintaining large-as-possible variance with respect to these generators.

As a specific example of this, we consider the three-parameter $\SU(2)$ estimation problem, where the goal is to estimate the three parameters coupling to $J_{x}$, $J_{y}$, and $J_{z}$ that uniquely identify a $\SU(2)$ operator \cite{PhysRevLett.116.030801}. 
This problem is analogous to a multiparameter problem where we encode information in all three directions of rotations. 
Previous work in this direction considered the compass state as an ideal candidate state for $\SU(2)$ parameter estimation \cite{PhysRevLett.116.030801}.
However, these states are not optimal in general representations of $\SU(2)$.
We approach this problem by using the notion of quantum metrology conditions to identify SU(2) representations \cite{serre1977linear} that possess optimal states, finding that while the compass state satisfies the quantum metrology condition in spin-$N/2$ representations with $N\equiv 0 \mod 8$, other representations contain optimal probe states which are not of the compass state form.
Our symmetry-based approach to satisfying the quantum metrology condition, which associates the quantum metrology condition to a trivial irrep of a finite group that contains the symmetry of the generator set, allows us to identify the number of optimal orthogonal states in a given representation.

One straightforward way to guarantee the satisfaction of our $\SU(2)$ quantum metrology condition is to instantiate any $\SU(2)$ representation as a tensor product and use a maximally entangled state.
Such an entanglement-assisted method was first demonstrated for spin-${1\over 2}$ representation in early work \cite{PhysRevA.65.012316}.
Extending this insight, we establish the general validity of utilizing entanglement for $\SU(2)$ parameter estimation.

In quantum metrology, finding an optimal measurement scheme is as important as finding an optimal state. It constitutes a challenge in multiparameter metrology due to the non-commuting generators involved.
However, it is possible to correct non-commuting errors in quantum error correction as long as the Knill-Laflamme conditions are satisfied. 
Using these insights one can use the quantum metrology conditions derived in this paper to find optimal measurement schemes for $\SU(2)$ parameter estimation which locally saturates the SLD quantum Cram\'{e}r-Rao bound.
Our measurement scheme works particularly well for small angles of rotation and can be straightforwardly extended to other multiparameter estimation schemes.

The remainder of the article is organized as follows.
In \cref{sec:Multi_parameter_estimation} we define the multiparameter estimation problem and derive the condition satisfied by optimal states for $\SU(2)$ multiparameter estimation.
In \cref{sec:SU_2_estimation}, we analyze in detail the structure of optimal states for the $\SU(2)$ estimation problem according to the symmetry group that guarantees their optimality.
In \cref{sec:entanglement_as_a_resource}, we explain how maximally entangled states of two isomorphic $\SU(2)$ representations provide optimal resources for the $\SU(2)$ estimation problem.
In  \cref{sec:optimality_as_a_function_of_theta}, we study the behavior of the quantum Fisher information matrix (QFIM) away from the small rotation limit and find the family of optimal states for this case.
In \cref{sec:measurement_that_saturates_QFI}, we use the quantum metrology condition to find a measurement that locally saturates the QFIM for $\SU(2)$ estimation. 
We conclude and explore possible future directions in \cref{sec:discussions_and_future_work}.
\vspace{0.7cm}

\section{SU(2) Estimation and the Quantum Metrology Conditions}
\label{sec:Multi_parameter_estimation}
In this section, we develop a framework that allows us to identify optimal probe states for the estimation of three real $\SU(2)$ parameters in any irreducible $\SU(2)$ representation. The outcome of the framework is a set of conditions, which turn out to be a special case of the Knill-Laflamme quantum error correction conditions, that characterize optimal probe states. Consider the group $\SU(2)$, of dimension $3$, and an irreducible representation $V$.
Our general multiparameter estimation problem is the $\SU(2)$ shift model over quantum states of $V$ \cite{holevo}
\begin{equation}
    \rho_{\vec{\theta}}= e^{-i\vec{\theta}\cdot \vec{J}}\rho e^{i\vec{\theta}\cdot \vec{J}}
    \label{eqn:probe1}
\end{equation}
where $\vec{J}=(J_{x},J_{y},J_{z})$ is a killing orthonormal basis of the Lie algebra $\mathfrak{su}(2)$ and $\vec{\theta}$ is the real vector of parameters to be estimated. 
Since we have chosen a representation $V$, $J_{j}$ should be interpreted as its corresponding matrix in that representation. Because we will assign equal cost to errors in each component of the estimate $\hat{\vec{\theta}}$ of $\vec{\theta}$ by using the deviation function 
\begin{equation}
    W_{\vec{\theta}}(\hat{\vec{\theta}}) = \Vert \hat{\vec{\theta}}-\vec{\theta} \Vert^{2},
    \label{eqn:dev}
\end{equation}
it would be natural to demand invariance of the probe state $\rho$ under operations that permute the generators.
% \jag{[Jonathan: I don't understand what's being said here.
% Evolving according to $s\vec{\theta}$ by definition leaves us with $\rho_{s\vec{\theta}}$.]}
However, the only symmetry of the generators that is implemented by $\SU(2)$ is the cyclic symmetry $Z_{3}$, so we demand that the probe state $\rho$ satisfies $C\rho C^{\dagger}=\rho$, where $C\in \SU(2)$ implements the cyclic permutation of $J_{x}$, $J_{y}$, $J_{z}$.
This minimal assumption can also be stated as $Z_{3}$ covariance of the model (\ref{eqn:probe1}), namely $\rho_{C^{-1}\vec{\theta}}=C\rho_{\theta}C^{\dagger}$, where $C$ acts on the components of $\theta$ in the same way as $\text{Ad}_{C}: \mathfrak{su}(2)\rightarrow \mathfrak{su}(2)$ acts on the components of $\vec{J}$.

The $Z_{3}$ covariance results in a substantial simplification of the SLD QFIM appearing in the quantum Cram\'{e}r-Rao bound. Recall that the parametrized state manifold $\lbrace \rho_{\vec{\theta}}\rbrace_{\vec{\theta}}$ has tangent space spanned by the symmetric logarithmic derivatives (SLD) $L_{\vec{\theta}}^{(j)}$, $j=1,\ldots,3$. The SLD QFIM $F(\vec{\theta})$ defines the metric on this manifold and has matrix elements given by \cite{holevo}:\begin{equation}
    F(\vec{\theta})_{i,j}:={1\over 2}\text{tr}\rho_{\vec{\theta}}[L^{(i)}_{\vec{\theta}} ,L^{(j)}_{\vec{\theta}}]_{+}
\end{equation}
 where the $[\cdot,\cdot]_{+}$ is the anti-commutator.  The quantum Cram\'{e}r-Rao bound (QCRB) based on symmetric logarithmic derivatives is given in matrix and scalar forms by
\begin{align}
\Sigma(\vec{\theta}) \ge F(\vec{\theta})^{-1} \label{eqn:qcrb} \\
E(W_{\vec{\theta}}(\hat{\vec{\theta}})) \ge \text{tr}[F(\vec{\theta})^{-1}] \label{eqn:trqcr}
\end{align}
where $\Sigma(\vec{\theta})$ is the covariance matrix $E(\vec{\Xi} \vec{\Xi}^{T})$ of the vector of deviations $\vec{\Xi}= (\hat{\theta}_{1}-\theta_{1},\hdots,\hat{\theta}_{3}-\theta_{3})^{T}$ obtained from a locally unbiased estimator $\hat{\vec{\theta}}$.
%-valued measurement that defines the expectation.
The inequality (\ref{eqn:trqcr}) is the scalar version of (\ref{eqn:qcrb}) obtained by taking the trace of both sides as $3 \times 3$ matrices \cite{PhysRevA.95.012305}. The minimal value of $\text{tr}[F^{-1}]$ is obtained on pure states, so the minimal QCRB is on the pure state manifold, justifying our subsequent restriction  to pure probe states $\rho$ in the rest of this work. With this restriction,
an explicit expression for the SLD operators is
\begin{equation}    L^{(j)}_{\vec{\theta}}:=2\partial_{\theta_{j}}\rho_{\vec{\theta}}.
\end{equation}
At $\vec{\theta}=\vec{0}$,
\begin{equation}
L^{(j)}_{\vec{0}} = -2iJ_{j}\rho_{\vec{0}} + h.c.
\label{eqn:sldzero_Ham}
\end{equation}
for $j=1,\ldots,3$. It is then straightforward to verify that the QFIM takes the following form at $\vec{\theta}=\vec{0}$ due to the $Z_{3}$ symmetry of $\rho$:
\begin{equation}
    F(\vec{0})=\begin{pmatrix}
        c & b &b\\
        b& c&b\\
        b&b&c 
    \end{pmatrix},
    \label{eqn:su2qfim}
\end{equation}
where $b\in \mathbb{R}_{+}$ and $c\in \mathbb{R}$.
The above matrix is a symmetric circulant matrix with eigenvalues
\begin{equation}
    \begin{aligned}
        \lambda_0&=c+2b,\\
        \lambda_i&=(c-b), \text{ for } 1\leq i\leq 2.
    \end{aligned}
\end{equation}
Thus one finds that,
\begin{equation}
    \mathrm{tr} (F(0)^{-1})=\frac{2}{c-b}+\frac{1}{c+2b}
\end{equation}
which is minimized for $b=0$ at the value ${3\over c}$, where
\begin{equation}
    c=4\text{Var}_{\ket{\psi}}J_{i} \; , \; i=1,\ldots, 3.
    \label{eqn:elel}
\end{equation}
We have therefore identified the structure of the QFIM of the general form (\ref{eqn:su2qfim}) which has the lowest value of $\text{tr}[F(0)^{-1}]$. 
By reading off the expectation values that define the matrix elements of 
\begin{equation}
\begin{aligned}
    F(0)_{i,j}&=2\left(\bra{\psi}J_i J_j\ket{\psi}+\bra{\psi}J_j J_i\ket{\psi}\right.\\   
    &\left.- 2\bra{\psi}J_i\ket{\psi}2\bra{\psi}J_j\ket{\psi} \right),
    \end{aligned}
\end{equation}
 this result provides us with a set of conditions characterizing the pure, $Z_{3}$ covariant probe states $\rho:=\ket{\psi}\bra{\psi}$ that obtain this minimal value in the spin-${N\over 2}$ representation
\begin{equation}
\begin{aligned}
    \langle J_i \rangle &=0\\
    \langle J_i J_l \rangle &=0 \, , \, i\neq l \\
      \langle J_i^2 \rangle &=\frac{{N\over 2}({N\over 2}+1)}{3},\\
      \label{eq:condition_theta_0}
\end{aligned}    
\end{equation}
% \begin{equation}
%     \begin{aligned}
%         \bra{\psi} J_{i} \ket{\psi}&=0\\
%         \bra{\psi} J_{i} J_{j}\ket{\psi}&=0\, , \, i\neq j\\
%          \bra{\psi} J_{i}^{2}\ket{\psi}&={c\over 4}.
%     \end{aligned}
%     \label{eqn:mmm}
% \end{equation}
%\jag{[Jonathan: Say that we used the Casimir.]}
where in the last condition, we combined the first condition with $Z_{3}$ covariance to uniquely determine the value of $c$ in (\ref{eqn:elel}) according to the $\mathfrak{su}(2)$ Casimir invariant.
The above conditions are a specific case of the Knill-Laflamme conditions of quantum error correction \cite{PhysRevA.55.900,bacon_2013}
\begin{equation}    \bra{\psi_{i}}E_a^{\dagger}E_b\ket{\psi_{j}}=C_{ab}\delta_{ij}
\label{eq:KL_condition}
\end{equation}
for error correction of a certain set of Kraus operators $\{E_a\}$, where $\lbrace \ket{\psi_{i}}\rbrace_{i}$ is an orthonormal basis of the code space. Specifically, assuming: 1. the set of Kraus operators to be $\{\mathds{1}\} \cup \{ J_{k} \}_{k}$, and 2. a code matrix $C=1\oplus {c\over 4}\mathbb{I}$, one finds that (\ref{eq:KL_condition}) reduces to \cref{eq:condition_theta_0}.
With these assumptions, the linear span of states satisfying the conditions \cref{eq:condition_theta_0} is a pure quantum error correction code (i.e., one for which the code space is contained in the kernel of every $E_{a}^{\dagger}E_{b}$ for which $a\neq b$ \cite{613213}), and the optimal multiparameter metrology condition \cref{eq:condition_theta_0} is succinctly written as
\begin{equation}\bra{\psi}E_a^{\dagger}E_b\ket{\psi}={c\over 4}\delta_{a,b}-\left({c\over 4}-1\right)\delta_{a,1}\delta_{b,1}\label{eq:metrology_condition}.
\end{equation}
%This idea is illustrated in \cref{fig:metrology_set_up}.
Thus the question of finding a probe state that exhibits the minimal value $\text{tr}[F(0)^{-1}]$ of the SLD QCRB over $Z_{3}$ covariant probe states amounts to finding a state that satisfies \cref{eq:metrology_condition}.
Note that in an SU(2) irrep $V$, the value of $c$ is determined by the value of the Casimir element $\sum_{i=1}^{3}J_{i}^{2}$ in $V$.
In Section \ref{sec:SU_2_estimation} we will use this fact to show that our restriction to $Z_{3}$ covariant probe states achieves the minimal value of $\text{tr}[F(0)^{-1}]$ overall probe states in the $\SU(2)$ irrep $V$.
% A condition sufficient but not necessary is that \cref{eq:KL_condition} equals zero when $E_a$ and $E_b$ are distinct, i.e., $C_{a,b}\propto \delta_{a,b}$. This indicates that each error transforms the initial state into orthogonal subspaces, enabling the recovery of the original state through projection onto these subspaces and forming good quantum codes. 
% Similarly, for the case of the multiparameter quantum metrology condition in \cref{eq:metrology_condition} implies that the generators $H_k$ take to orthogonal subspaces. 
%  While the metrology condition imposes a stricter requirement—demanding that all generators $H_k$ transform the codewords to have the same scalar product, however, for good quantum codes in error correction, the error operators $E_a$ may transform codewords to different scalar products.

Note that in the multiparameter estimation setting, the SLD QCRB is not guaranteed to be satisfied by any measurement for extracting an estimate $\hat{\vec{\theta}}$. 
For the existence of a measurement that saturates (\ref{eqn:qcrb}) at a given $\vec{\theta}$, it is necessary and sufficient that the SLD operators locally commute \cite{matsu}.
For a pure probe $\rho_{\vec{\theta}}=\ket{\psi(\vec{\theta})}\bra{\psi(\vec{\theta})}$, satisfaction of the weaker conditions
\begin{align}
&{} \text{1. } \text{tr}\rho_{\vec{\theta}}[L^{(i)}_{\vec{\theta}},L^{(j)}_{\vec{\theta}}] =0 \label{eqn:cond1} \\
&{} \text{2. } F(\vec{\theta}) \text{ invertible }
\label{eqn:cond2}
\end{align}
is necessary and sufficient for the existence of a measurement that saturates (\ref{eqn:qcrb}) in a single shot \cite{PhysRevA.94.052108}.
For pure probes, the condition (\ref{eqn:cond1}) is equivalent to the condition that $\text{Im}\left( \partial_{\theta_{i}}\ket{\psi(\vec{\theta})},
\partial_{\theta_{j}}\ket{\psi(\vec{\theta})} \right)=0$ for all $i,j$ \cite{PhysRevLett.119.130504}. At $\vec{\theta}=0$, a probe state satisfying \cref{eq:condition_theta_0} also satisfies (\ref{eqn:cond1}), and (\ref{eqn:cond2}) is then satisfied for any $c>0$.

Lastly, we outline here the strategy that we implement in Section \ref{sec:SU_2_estimation} to identify simple probe states that satisfy \cref{eq:condition_theta_0} in the irrep $V$.
% \jag{[Jonathan: I'm deleting comments about the largest value of $c$ since it's fixed by the Casimir operator.]}
The method consists of identifying a small finite subgroup $G$ of $\SU(2)$ that possesses two properties: 1. it contains the $Z_{3}$ subgroup generated by $C$, 2. its trivial representation contained in the irrep $V$ of $\SU(2)$ consists of states that satisfy \cref{eq:condition_theta_0}. A probe state $\ket{\psi}$ taken from that trivial representation is then guaranteed to achieve the minimal value of $\text{tr}[F(0)^{-1}]$, which we use as our quantifier of optimality.
In the next section we look at two examples of such $G$, show how the optimal multiparameter metrology conditions are satisfied in the trivial irreps of $G$, and compare optimal states obtained from different choices of $G$.

% \begin{figure}
%     \centering
%     \includegraphics[width=\columnwidth]{Fig_set_up.eps}
%     \caption{Setting of quantum error correction inspired multiparameter quantum metrology. 
%     We have the standard quantum error correction scheme for encoding a qubit, where we have a set of Kraus operators $\{E_a\}$, and for encoding of a qubit, one needs to satisfy the Knill-Laflamme conditions. 
%     For quantum metrology, the goal is to find a state $\ket{\psi}$ and we have the set of operators $E_a=\{\mathds{1}, H_l\}$, where $H_l$ is the set of Hamiltonians in which we encode our parameter to estimate into and for this setting one needs to satisfy a new condition.}
%     \label{fig:metrology_set_up}
% \end{figure}

Although our main focus in this work is on the estimation of $\SU(2)$, one can find similar examples with more generators in which quantum metrology conditions analogous to \cref{eq:condition_theta_0} allow one to identify optimal probe states in a representation.
In \cref{sec:qudits_d_4}, we discuss a four-parameter estimation problem in which the parameters are coupled to a subset of generalized Gell-Mann matrices in the Lie algebra $\mathfrak{su}(4)$. There we implement the same general strategy for identifying optimal probe states by using appropriately generalized quantum metrology conditions.

\section{Bounds on SU(2) estimation and optimal states}
\label{sec:SU_2_estimation}
We now consider our main problem of $\mathrm{SU}(2)$ estimation in a specific representation.
This problem can be alternatively phrased as the task of estimating all three components of a magnetic field using two-level systems.
The Hamiltonian for this case is,
 \begin{align}
    H(\vec{\theta})&=\vec{\theta}\cdot \vec{J}
    \label{eqn:ooo}
\end{align}
where $[J_{i},J_{j}]=i\epsilon_{ijk}J_{k}$ is the $\mathfrak{su}(2)$ algebra.

We will use the notation $\text{tr}[(\cdot)\rho_{\vec{\theta}}]$ and $\langle \cdot \rangle_{\rho_{\vec{\theta}}}$ interchangeably for expectation values, and write the vector of spin operators $\vec{J}=(J_{1},J_{2},J_{3})=(J_{x},J_{y},J_{z})$. 
% Note that $U(\vec{\theta})^{\dagger}\vec{J}U(\vec{\theta}) = \vec{J}R(\vec{\theta})$ for a $3\times 3$ orthogonal matrix $R(\vec{\theta})$. 
The three parameters $\vec{\theta}$ are imprinted on a probe state $\rho$ according to $\rho_{\vec{\theta}}=U(\vec{\theta})\rho U(\vec{\theta})^{\dagger}$, where   $U(\vec{\theta}):=\exp(-i H(\vec{\theta}))$ and $\rho$ is a quantum state in a spin-$N/2$ representation of $\SU(2)$, which can be interpreted as a symmetric state of $N$ qubits. For simple probe states $\rho$ such as Dicke states, the Riemannian geometry of the pure state manifold defined by (\ref{eqn:ooo}) can be analyzed using a canonical method \cite{PhysRevA.48.4102}. For pure $\rho$, the symmetric logarithmic derivatives (SLD) of the problem are given by
\begin{align}
L^{(j)}_{\vec{\theta}}&:=2\partial_{\theta_{j}}\rho_{\vec{\theta}} \nonumber \\
&= 2i\left[ {\sin \Vert \vec{\theta}\Vert\over \Vert \vec{\theta}\Vert}  J_{j} +\left(1-{\sin\Vert\vec{\theta}\Vert\over \Vert\vec{\theta}\Vert} \right) {\theta_{j}\over \Vert \vec{\theta}\Vert^{2}}\vec{\theta}\cdot \vec{J} \right. \nonumber \\
&{} \left. - 2{\sin^{2}{ \Vert \vec{\theta}\Vert\over 2}\over \Vert \vec{\theta}\Vert^{2}}(\vec{e}_{j}\times \vec{\theta})\cdot \vec{J} \right] \rho_{\vec{\theta}} + h.c.
\label{eqn:sld}
\end{align}
where $\vec{e}_{j}$ is the unit vector in direction $j$. See  \cref{sec:app1} for a derivation of (\ref{eqn:sld}). 
At $\vec{\theta}=0$,
\begin{equation}
L^{(j)}_{0} = -2iJ_{j}\rho + h.c.
\label{eqn:sldzero}
\end{equation}

The SLD QFIM $F(\vec{\theta})$ has matrix elements $F(\vec{\theta})_{i,j}:={1\over 2}\text{tr}\rho_{\vec{\theta}}[L^{(i)}_{\vec{\theta}} ,L^{(j)}_{\vec{\theta}}]_{+}$ \cite{holevo}. 
From the fact that the harmonic mean is a lower bound on the arithmetic mean, one finds that (suppressing $\vec{\theta}$)
\begin{align}\text{tr}F(\vec{\theta})^{-1}&\ge {3\over \text{Tr}F(\vec{\theta})}\nonumber \\
 &\ge {9\over 4\sum_{i=1}^{3}\text{Var}_{\rho_{\vec{\theta}}}J_{i}} \nonumber \\
    &\ge {9\over N^{2}+2N}
\label{eqn:aaa}
\end{align}
where  $\text{Tr}:={1\over N+1}\text{tr}$ is the normalized trace in the symmetric subspace and the last line uses the Casimir expectation value $\sum_{i=1}^{3}\langle J_{i}^{2}\rangle={N\over 2}\left({N\over 2}+1\right)$ to bound the variance.
This last inequality is saturated by taking any pure state with $\langle \vec{J}\rangle=0$.
We can in fact saturate all inequalities in \cref{eqn:aaa} by using a state satisfying the quantum metrology conditions given in \cref{eq:condition_theta_0}.

Note that for a symmetric state $\rho$, the QFIM depends only on the one-qubit and two-qubit reduced states \cite{PhysRevLett.116.030801}
\begin{align}
        F(0)_{ij}&= N(N-1) \text{tr}\left[\rho^{(2)}\left( \sigma_{i}\otimes \sigma_{j} \right)\right] \nonumber \\
        &- N^{2}\text{tr}\left[\rho^{(1)}\sigma_{i}\right]\text{tr}\left[ \rho^{(1)}\sigma_{j}\right]+N\delta_{i,j}
    \end{align}
where $\sigma_{i}$, $i=1,2,3$, are the Pauli matrices.
Satisfaction of the quantum metrology conditions  \cref{eq:condition_theta_0} in the spin-$N/2$ is expressed via the conditions
\begin{equation}
     \text{tr}\left[\rho^{(2)} (\sigma_{i}\otimes \sigma_{j}) \right] = \begin{cases}0 &  i\neq j  \\
    {1\over 3} & i=j \end{cases}.
    \label{eqn:utut}
\end{equation}
In Section \ref{sec:optimality_as_a_function_of_theta} we derive the two-qubit reduced states $\rho^{(2)}$ for optimal probe states.

\subsection{Optimal state for SU(2) estimation}
\label{sec:Finding_optimal_states}

In this section, we focus on finding states that satisfy \cref{eq:condition_theta_0} for spin $J$ being an integer, i.e., $N$ being an even number of qubits. 
To find these states we rely on the fact that restricting an irreducible representation of $\SU(2)$ to a representation of one of its finite subgroups $G$ gives, in general, a reducible representation of $G$.
Identifying the irreps of $G$ contained in the restricted representation and their multiplicities allows us to classify the $G$ symmetry of states that transform under $\SU(2)$.
In particular, $2$-dimensional irreps of some finite groups provide candidates for encoding logical qubits \cite{Gross2021,omanakuttan2023multispin}.

In the present context, one way to obtain a state that satisfies the quantum metrology conditions is to identify a finite group $G$ such that states in its trivial irrep satisfy the conditions.
A trivial irrep of $G$ can be obtained by restricting the spin-$N/2$ irreducible representation of $\SU(2)$ to $G$ (so that $g \in G$ is a unitary $N+1\times N+1$ matrix) and taking the image of the projection
\begin{equation}
    \Pi=\frac{1}{\lvert G \rvert} \sum_{g\in G} g.
    \label{eq:trivial-irrep-projector}
\end{equation}
The $G$-invariant component of a given $N$ qubit symmetric state $\ket{\psi}$ is obtained from $\Pi\ket{\psi}$.

As a first example of $G$, one can use the binary tetrahedral group ($2T$) \cite{Gross2021,omanakuttan2023multispin}.
Trivial irreps do not exist in irreducible representation of $\SU(2)$ where $J$ is a half-integer nor in the integer cases $J=1,2$, and  $5$ \cite{Gross2021}. 
The $2T$ group is a double cover of the alternating group $A_{4}$, the latter
%\jag{(Be careful here; the action of 2T on integer irreps of SU(2) is $A_4$, but 2T is a different group than $A_4$)},
having order 12 and  defined according to the generators and relations
\begin{equation}
    A_{4}=\langle G_{1},G_{2}\vert G_{1}^{3}=1,G_{2}^{2}=1,(G_{1}G_{2})^{3}=1\rangle
\end{equation}
where for integer representations of $\SU(2)$ one can choose the orientation
\begin{align}
    G_1&=\exp \left(\frac{2\pi i}{3}\frac{J_x+J_y+J_z}{\sqrt{3}}\right) \nonumber \\
    G_{2}&= e^{-i\pi J_{z}}
    \,,
\end{align}
corresponding to the symmetries of a tetrahedron inscribed in a cube with faces perpendicular to the coordinate axes.
From the form of the generators, one can see that $A_{4}\cong $2T$/Z_{2}$ is faithfully represented for integer $J$. Note that $G_{1}$ implements the basic $Z_{3}$ symmetry of the generators, which is the starting point of our multiparameter metrology setting.

Since $2T$ has symmetries taking $J_i$ to $-J_i$, the first moments of the generators $\langle J_i\rangle$ must vanish. Further, because $2T$ contains the $Z_{3}$ symmetry of the generators,  the $\langle J_i^2\rangle$ terms in the quantum metrology conditions of \cref{eq:condition_theta_0} must be equal to one another, and therefore take the optimal values because of the Casimir operator. Finally, for any $J_i$ and $J_j$, there is a $2T$ symmetry taking $J_i$ to $-J_i$ while leaving $J_j$ unchanged, so the $\langle J_iJ_j\rangle$ terms also must vanish.
Therefore states in the trivial irrep of $2T$ automatically satisfy \cref{eq:condition_theta_0}.

For example, for $J=3$ the multiplicity of the trivial irrep of $2T$ occurs with multiplicity 1, so the optimal state, written in the $\ket{J,m_z}$ basis, is
\begin{equation}
    {\ket{J=3,m_z=2}+ \ket{J=3,m_z=-2} \over \sqrt{2}}.
    \label{eq:2T_state_j_3}
\end{equation}
Moving to general spin-$J$ representations, it is well known that in the case of single-parameter estimation, an optimal state is a Greenberger-Horne-Zeilinger (GHZ) state. For example, when the parametrized unitary is $e^{-i\theta J_z}$, an optimal state is 
\begin{equation}
    \ket{\text{GHZ}_{z}}:={\ket{J,m_z=J}+\ket{J,m_z=-J}\over \sqrt{2}}.
\end{equation}
Motivated by the optimality of the GHZ state for estimation of rotation about a single axis \cite{Braunstein_caves_metrology}, a family of states was introduced for the three-parameter SU(2) estimation problem 
\begin{equation}
    \ket{\mathrm{compass}}=\mathcal{N} \sum_{\ell = x,y,z} e^{i\delta_\ell}\ket{\text{GHZ}_\ell},
    \label{eq:compass_state}
\end{equation}
where $\mathcal{N}$ is the normalization constant and $\{\delta_\ell\}_{\ell=x,y,z}$ are real, adjustable phases.
However, as we will show, the optimality of the compass state for $N\equiv 0 \text{ mod } 8$ can be viewed as a consequence of its $2T$ invariance for these representations.
% \jag{[Jonathan: I don't think this is what the following text proves...
% We only show that the compass state is a $2T$ state when $N\equiv0\mod8$; we don't demonstrate anything about lack of optimality in the other cases, unless I'm missing something?]}
This fact raises the question of whether these optimal compass states are $2T$ invariant because such invariance is sufficient for optimality.
For spin-$N/2$ $\SU(2)$ representations we plot the overlap of compass states in \cref{eq:compass_state} with the trivial irrep of $2T$ in \cref{fig:compare_2T_compass}.
The ${2T}$-invariance of the compass state is observed for $N\equiv 0 \mod 8$, in agreement with the known optimality result.

\begin{figure}[t]
    \includegraphics[width=\columnwidth]{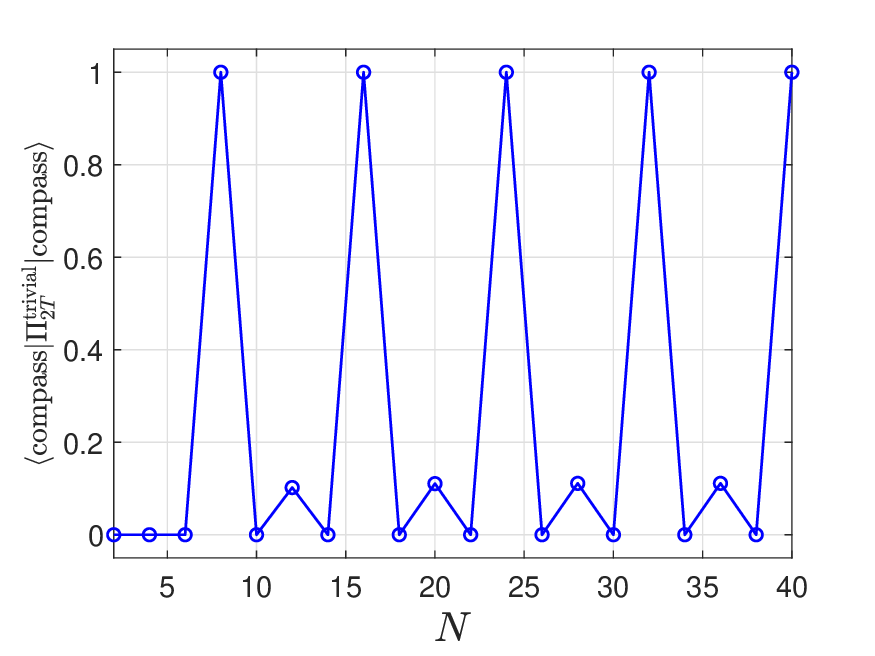}
    \caption{Component of the compass state defined by \cref{eq:compass_state} in trivial irreps $2T$ in spin-$N/2$ representations of $\SU(2)$. 
    As seen from the figure, there is a periodic behavior when a compass state becomes $2T$, and the periodicity is when $N \text{ mod } 8=0$.
    }
    \label{fig:compare_2T_compass}
\end{figure}
To prove the $N\equiv 0 \mod 8$ 
%\jag{[Jonathan: our mod notation is all over the place.
%I would standardize it to be like here.]} 
periodicity of $2T$ invariance of the compass state, one can look at the action of specific elements of the $2T$ group on the highest weight state $\ket{J,m_z={N\over 2}}$.
Consider, e.g., $G_2=e^{-i\pi J_z}$, and note that
\begin{equation}
\begin{aligned}
     &\bra{J,m_z={N\over 2}}G_2\ket{J,m_z={N\over 2}}\\
     &=\bra{0}^{\otimes N} \bigotimes_{i=1}^{N} \exp(-i\frac{\pi}{2}\sigma_z^{(i)}) \ket{0}^{\otimes N}=(-i)^N.
\end{aligned}
\end{equation}
Therefore, when $N \text{ mod } 4=0$, invariance with respect to this group generator is obtained. 
Next, one can consider the other generator $G_{1}$ which yields
\begin{multline}
    \bra{J,m_z={N\over 2}}G_1\ket{J,m_z={N\over 2}}=
    \\
    \bra{0}^{\otimes N} \bigotimes_{i=1}^{N} \exp(-i\frac{2\pi}{6}\left(\sigma_x^{(i)}+\sigma_y^{(i)}+\sigma_z^{(i)}\right)) \ket{0}^{\otimes N}
    \\
    =\left(\frac{1+i}{2} \right)^{N} 
    =\frac{1}{2^{N/2}}\exp\left(-i\frac{N\pi}{4}\right)
    \,,
\end{multline}
so the above matrix element is 1 if and only if $N \text{ mod }8 =0$.
For example, for $J=4$ the trivial irrep of $2T$ is one-dimensional, so the unique optimal state with $2T$ symmetry is indeed the compass state
\begin{equation}
\begin{aligned}
    \sqrt{\frac{5}{24}}\left(\ket{m_z=4}+\ket{m_z=-4}\right)+\sqrt{\frac{7}{12}}\ket{m_z=0}.
\end{aligned}
    \end{equation}  
By contrast, the optimal binary tetrahedral state and compass state are different for $J=3$.
We compare their spin-Wigner functions~\cite{stratonovich1957distributions,spin_wigner_dowling,spin_wigner_moyal} in \cref{fig:compare_2T_compass_husimi}.

\begin{figure*}
\centering
 \subfloat[$2T$]{\includegraphics[width =\columnwidth]{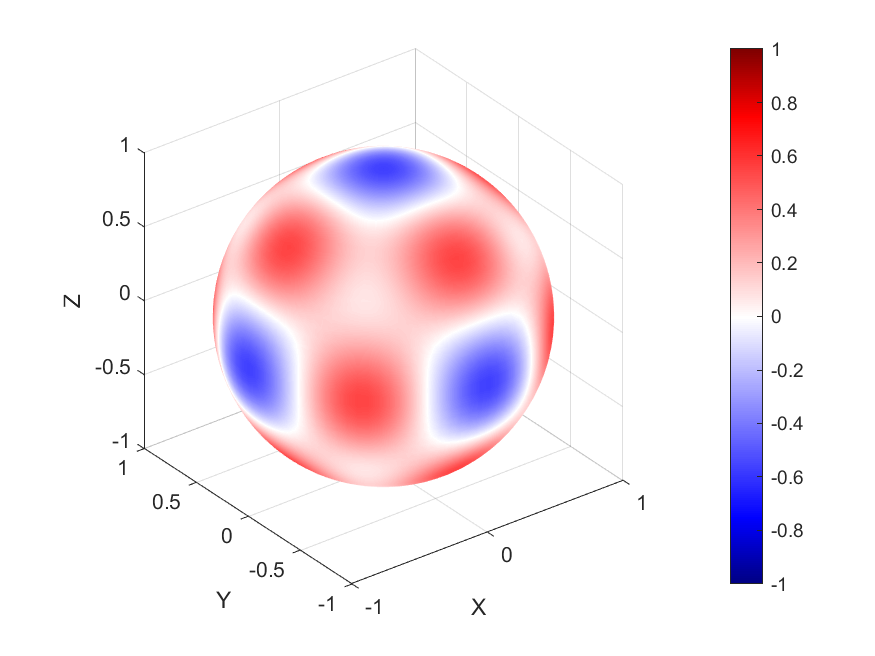} \label{fig:Fig_2T}}
    \subfloat[Compass]{\includegraphics[width =\columnwidth]{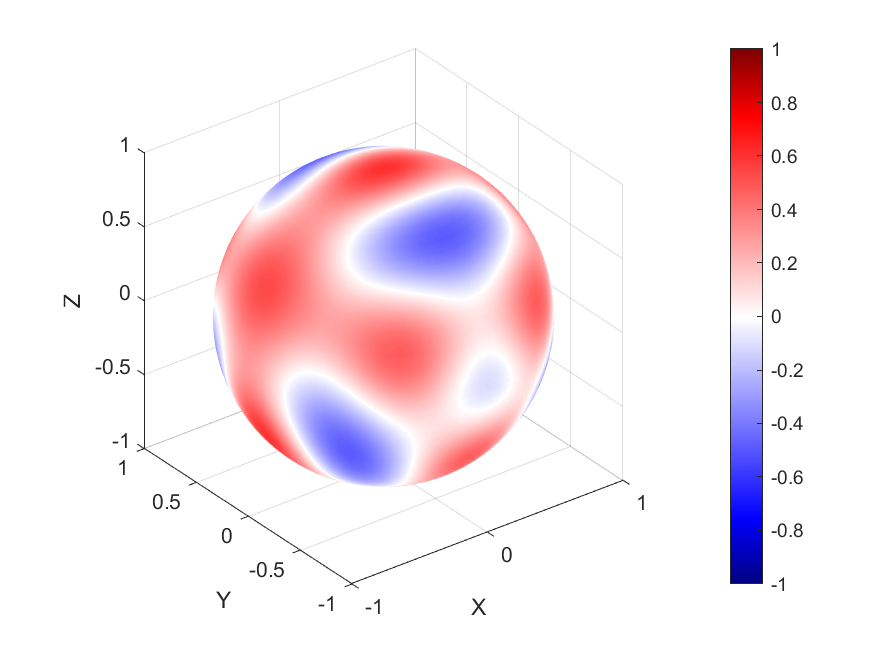} \label{fig:Fig_compass}}  
    \caption{Spin-Wigner function for the $2T$ invariant and compass state for $J=3$. (a) Wigner function of the state in the trivial irrep of $2T$ as given by \cref{eq:2T_state_j_3}.
    (b) Wigner function of the compass state for $J=3$, which is $\ket{\mathrm{compass}}=1/\sqrt{2}\left(e^{-0.1476 \pi i }\sqrt{\frac{5}{12}}\ket{J_z=3}+e^{-i\pi/2}\sqrt{\frac{5}{4}}\ket{J_z=-1}+\sqrt{\frac{1}{3}}\ket{J_z=-3}\right)$.
    % \jag{[Jonathan: these axes claim to be the same, but the state looks different.
    % I feel like there should be an X somewhere...
    % Figure (b) looks different since I wrote the previous sentence, but now the axes look like they're not aligned with the page.
    % Figure (a) also looks no different from the two perspectives, so I'm not convinced that it helps to show all these.
    % I think you should choose the same orientation for both Wigner functions, and then have three panels: the 2T state, the compass state, and the axes.
    % If you do the axes in something like an orthographic projection, then you can plot the X Y and Z axes in 3d as well, with the same code, to avoid these mismatches we've been struggling with.] 
    }
    \label{fig:compare_2T_compass_husimi}
    \end{figure*}

% Discussion of states which are $A_{4}$ symmetric appears in Section \ref{sec:comp}.

In the symmetry analysis of $2T$, all quantum metrology conditions \cref{eq:condition_theta_0} are satisfied.
However, it is possible to combine invariance under a smaller finite group $G$ whose trivial irreps do not guarantee all the quantum metrology conditions with a fine-tuning method to satisfy the remaining quantum metrology conditions.
Consider the $S_3$ group, the symmetric group on 3 letters, having order $\vert G\vert = 6$. 
This is a smaller group compared to $A_4$ and does not imply equal variances of the $J_{x}$, $J_{y}$, $J_{z}$ generators.
As a finite subgroup of $\SU(2)$, the group $S_3$ has the following presentation 
\begin{align}
    G_1&=\exp \left(\frac{2\pi i}{3}J_{z}\right) \nonumber \\
    G_{2}&= e^{i\pi J_{x}}\nonumber \\
    S_{3}&=\langle G_{1},G_{2}\vert G_{1}^{3}=1, G_{2}^{2}=1, (G_{2}G_{1})^{2}=1\rangle
    \,,
\end{align}
Up to normalization, an $S_{3}$ invariant state is obtained by applying the projector onto the trivial irrep defined in \cref{eq:trivial-irrep-projector}
% the essentially idempotent operator (i.e., the $S_{3}$ twirl)
% \begin{equation}
%     \sum_{g\in S_{3}}g = 1+G_{1}+G_{2}+G_{1}^{2}+G_{2}G_{1}+G_{1}G_{2}
% \end{equation}
to a chosen state.
For example, one can consider an $\SU(2)$ coherent state $\ket{\zeta}$ defined as \cite{arecchi, lieb}
\begin{equation}
    \ket{\zeta} = {1\over \left(1+\vert \zeta\vert^{2}\right)^{N\over 2}}\left( \ket{0}+\zeta\ket{1}\right)^{\otimes N} \propto e^{\zeta J_{-}}\ket{J,m_{z}=J}
\end{equation} 
with $\zeta\in \mathbb{C} \cup \lbrace \infty \rbrace$ corresponding to a point on the 2-sphere by the stereographic projection from the south pole.
Taking $\zeta = \tan {\xi \over 2}$ with $\xi\in [0,\pi/2)$, applying the $S_{3}$ twirl results in the following superposition of $\SU(2)$ coherent states on a triangular prism:
\begin{multline}
     \ket{\psi(\xi)}\propto
     \ket{\tan {\xi \over 2}}
     +e^{\pi i N\over 3}\ket{e^{-2\pi i /3}\tan {\xi \over 2}}
     \\
     +e^{2\pi i N\over 3}\ket{e^{-4\pi i /3}\tan {\xi \over 2}}
     +e^{i \pi N\over 2}\ket{\cot {\xi \over 2}}
     \\
     +e^{\pi i N\over 6}\ket{e^{2\pi i /3}\cot {\xi \over 2}}
     +e^{2\pi i N\over 3}\ket{e^{4\pi i /3}\cot {\xi \over 2}}
     \label{eqn:tp}
\end{multline}
In the limit $N\rightarrow \infty$, one finds that fine-tuning the parameter $\xi$ to $\xi = \cos^{-1}{1\over \sqrt{3}}$, then $\langle J_i^2\rangle_{\ket{\psi(\xi)}}={1\over 3}{N\over 2}({N\over 2}+1)$ for $N\equiv 0 \mod 2$ and all $N>6$.
The tradeoff for reducing the size of the symmetry group is that a parameter must be fine-tuned to satisfy the quantum metrology conditions.
Enforcing instead the $2T$ symmetry allows one to obtain optimal states without fine-tuning due to all quantum metrology conditions being satisfied.
\begin{figure*}
\centering
  \subfloat[]{\includegraphics[width =1.15\columnwidth]{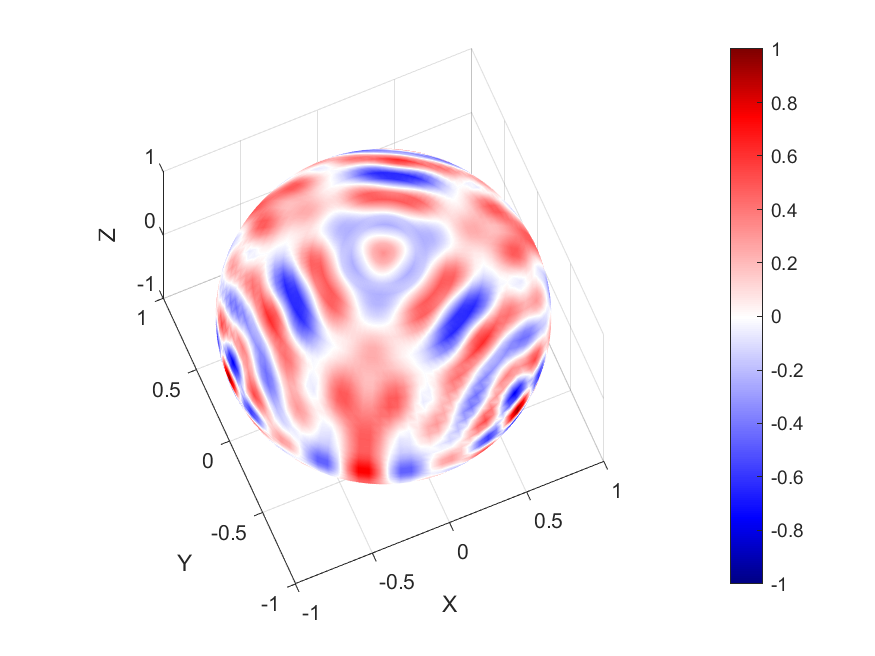} \label{fig:Fig_s3_a}}
    \subfloat[]{\includegraphics[width =0.85\columnwidth]{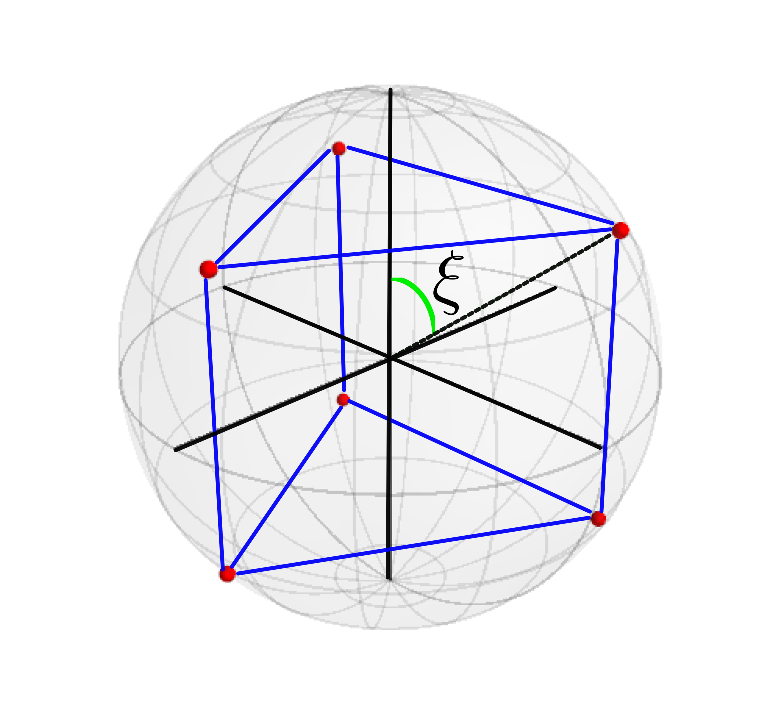} \label{fig:Fig_s3_b}}  
    \caption{(a) shows the spin-Wigner function of the state in the trivial irrep of $S_3$ for $J=10$.
    The triangular prism nature of the states is visible and becomes more pronounced for larger $J$.
    (c) In the $J_{z}$ basis, the state is a superposition of six $\SU(2)$ coherent states arranged on the vertices of a triangular prism. One can optimize the angle $\xi$ to find a state satisfying the quantum metrology conditions in \cref{eq:condition_theta_0} (see main text for more details). The specific state shown has $\xi=\cos^{-1}{1\over \sqrt{3}}$ in (\ref{eqn:tp}).
    % \jag{[I also don't think much is gained by showing from both the +X and -X directions (since it is symmetric under this transformation).
    % I'd maybe do like I suggest above and use an orthographic projection, like you already have for part (b), for your single perspective.]}
    }
    \label{fig:fig_s3_triangular}
\end{figure*}
In \cref{fig:fig_s3_triangular} (a) and (b) we plot the spin-Wigner function for $\ket{\psi(\xi)}$ for $J=10$, depicting the triangular prism structure of the states.
% By optimization of the angle $\xi$, a state is obtained that satisfies the quantum metrology conditions in \cref{eq:condition_theta_0}. 

Another interpretation of the fine-tuning of $S_3$-invariant states that leads to $A=B=C$ (with $A=\langle J_x^2\rangle$, $B=\langle J_y^2\rangle$, and  $C=\langle J_z^2\rangle$)  is provided by the observation that the parameter $\xi$ in (\ref{eqn:tp}) is actually just a special choice of the initial state to which the $S_{3}$ twirl is applied. 
Therefore, one expects that the fine-tuning can also be implemented simply by taking a superposition of states in trivial irreps of $S_{3}$ which are associated with different values of $A$ and $C$.
To find the multiplicity of the trivial irrep, one considers the projector $\Pi_{\mathrm{trivial}}= \frac{1}{6}\sum_{g\in S_{3}}g$ and calculates its trace in the spin-$J$ representation 
 \begin{equation}
     \mathrm{tr}\left(\Pi_{\mathrm{trivial}}\right)=\frac{1}{6}\left[2J+1+3(-1)^J+2\frac{\sin \left(\frac{\pi}{3}(2J+1)\right)}{\sin \left(\frac{\pi}{3}\right)}\right].
 \end{equation}
In \cref{fig:number_of_projector_s3}, the multiplicity of the trivial irrep is shown as a function of $J$.
From the figure, it is clear above a certain value of $J$ the multiplicity is more than one, and one can therefore construct a state
\begin{equation}
    \ket{\phi}:=\alpha \ket{\psi_0}+\beta\ket{\psi_1},
    \label{eq:find_optimal_state_S3}
\end{equation}
where $\ket{\psi_i}$ are orthogonal, $S_{3}$-invariant states.
Optimizing over $\alpha,\beta$ allows one to obtain $A=B=C$.
\begin{figure}
    \centering
    \includegraphics[width=\columnwidth]{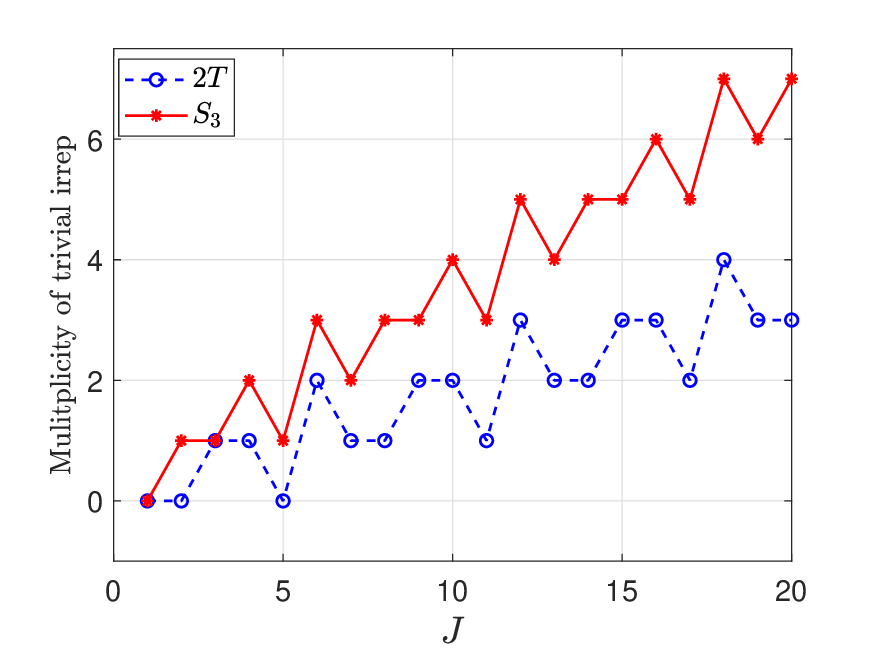}
    \caption{Multiplicity of the trivial representation of $S_3$ ($2T$) as a function of $J$ (we only consider the case of integral $J$).
From the figure, it is clear after a certain value of $J$, the multiplicity is more than one which one can use to satisfy the quantum metrology conditions given in \cref{eq:condition_theta_0}.
However, the trivial irrep of the $2T$  group naturally satisfies the quantum metrology conditions for $\SU(2)$ metrology as long as the multiplicity of the trivial irrep is larger than or equal to $1$. }
    \label{fig:number_of_projector_s3}
\end{figure}
For example, for $J=4$, the $S_{3}$ trivial irrep has  multiplicity $2$, with orthogonal states given as
\begin{equation}
\begin{aligned}
\ket{\psi_0}&=\sqrt{\frac{42}{1121}}\left(\ket{m_z=3}+\ket{m_z=-3}\right)\\
    &+e^{i\phi}\sqrt{\frac{1037}{1121}}\ket{m_z=0},\\
    \ket{\psi_1}&=\sqrt{\frac{1037}{2242}}\left(\ket{m_z=3}+\ket{m_z=-3}\right)\\
    &+\sqrt{\frac{84}{1121}}e^{ i\left(\pi +\phi\right)}\ket{m_z=0},\\
\end{aligned}    
\end{equation}
where $\phi=0.4559\pi$.
% \jag{(Siva will remind himself how he found these states, and why they're chosen to not be orthogonal.)}
% \SO{I was making the states orthogonal using the QR decomposition, however, there was a typo in the previous version and I have fixed it now.}
% \jag{[Jonathan: These are not orthogonal as is.
% I suspect at least part of the problem is the $\pi$ in front of $\phi_1$ and $\phi_2$ in the last term.]}
This superposition can be fine-tuned with $\alpha=0.7251$ and $\beta=0.6891$ to achieve a state 
\begin{equation}
    \sqrt{\frac{10}{27}}\left(\ket{m_z=3}+\ket{m_z=-3}\right)+e^{0.4282\pi i}\sqrt{\frac{7}{27}}\ket{m_z=0}
\end{equation}
that satisfies the quantum metrology conditions.

% \section{ Binary Tetrahedral (2T) states for $\mathrm{SU}(2)$ quantum sensing }
% One can use tetrahedral trivial irreps as the sensing state for $\{J_x,J_y,J_z\}$ rotation:
% \begin{equation}
%     U\ket{\psi}_0=\ket{\psi}_0  \text{ for } U \in \mathrm{2T}.
% \end{equation}
% The quantum Fisher information for $J_x$ is 
% \begin{equation}
%     Q_{xx}= \tr(\rho L_x^2)
% \end{equation}

\section{Entanglement-assisted SU(2) estimation}
\label{sec:entanglement_as_a_resource}
To this point, our model for probe states for the $\SU(2)$ estimation problem has been the symmetric subspace of $N$ two-level systems. However, such a restriction is not inherent in the quantum metrology conditions, and replacing the $\mathfrak{su}(2)$ generators by $J_{i}\mapsto J_{i}\otimes \mathbb{I}$ allows us to consider probe states in the tensor product of two $\SU(2)$ representations. 
The motivation in Ref.~\cite{PhysRevA.65.012316} for considering maximally entangled qubit states for $\SU(2)$ estimation comes from the dense coding protocol for qubits, in which an entangled resource state allows one to encode two bits of information using local operations instead of just one as required by the single-system Holevo bound. 
Similarly, noting that
\begin{align}
    &{}\int dU \, U\ket{m_z = J}\bra{m_z = J} U^{\dagger} \nonumber \\
    &{}= \sum_{m_{z}',m_{z}''}\int dU D_{m_{z}',J}(U)\overline{D_{m_{z}'',J}(U)} \ket{m_z'}\bra{m_z''} \nonumber \\
    &{}= {1\over 2J +1} \sum_{m_{z}=-J}^{J}\ket{m_{z}}\bra{m_{z}}
\end{align}
ones finds that the ensemble $\lbrace U|m_z{=}J\rangle\langle m_z{=}J|U^{\dagger}\rbrace$ with $U$ uniform in SU(2) achieves the maximum Holevo information $\log(2J+1)$ possible in a spin-$J$ representation, whereas carrying out the analogous calculation with
the entangled state ensemble $\lbrace U\otimes \mathbb{I}\ket{\psi}\bra{\psi}U^{\dagger}\otimes \mathbb{I}\rbrace$ with $U$ uniform in SU(2) and $\ket{\psi}$ the maximally entangled state in  (\ref{eqn:maxen}) below achieves the maximum Holevo information $2\log (2J+1)$ possible in a tensor product of spin-$J$ representations.
With the same motivation as Ref.~\cite{PhysRevA.65.012316}, we now consider the entanglement-assisted $\SU(2)$ estimation setting in which a maximally entangled state of two spin-$J$ representations serves as a probe state for entanglement-assisted SU(2) estimation, i.e., for estimating $\vec{\theta}$ in $U(\vec{\theta})\otimes \mathbb{I}$ with general probe states in the tensor-product space.
Recall that the entanglement-assisted scheme for one-parameter quantum estimation with pure state probes exhibits no advantage for the same number of calls to the parametrization \cite{PhysRevLett.96.010401} (in stark contrast with noisy probe states \cite{PhysRevA.97.042112,PhysRevLett.113.250801,PhysRevA.94.012101}).
However, the optimal spin-$J$ states described in \cref{sec:SU_2_estimation} require symmetry properties that may be challenging to generate and, further, do not exist in every spin-$J$ representation. 
Here we show that the quantum-metrology conditions can be satisfied by a class of maximally entangled states. The form of the state is motivated by the optimal two-qubit state for local SU(2) estimation derived by Fujiwara \cite{PhysRevA.65.012316}, which in turn was motivated by the dense coding approach to Hamiltonian distinguishability \cite{childspresk}.
Specifically, consider the following state in the tensor product of spin-$J$ representations
\begin{equation}
\begin{aligned}
     \ket{\psi}&=\frac{1}{\sqrt{2J+1}}\sum_{i=-J}^J\ket{J,m_z=i}\ket{J,m_z=i}.
\end{aligned}
\label{eqn:maxen}
\end{equation}
 Invoking the fact that you have a maximally entangled state, the expectation value of these single-system observables is proportional to their trace given as,
\begin{equation}
    \bra{\psi} A\otimes \mathds{1}\ket{\psi}=\frac{1}{2J+1}\mathrm{tr}\left(A\right).
\end{equation}
Thus we get for the state in \cref{eqn:maxen},
\begin{equation}
    \begin{aligned}
        \langle J_i \otimes \mathds{1} \rangle &=0\\
    \langle J_i J_l  \otimes \mathds{1} \rangle &=0 \, , \, i\neq l \\
      \langle J_i^2  \otimes \mathds{1}  \rangle &=\frac{{J}({J}+1)}{3}
    \end{aligned}
\end{equation}
for $i,j\in \{x,y,z\}$.
Thus one can conclude that the conditions
\cref{eq:condition_theta_0} can be satisfied in a tensor product of any spin-$J$ representations by using a maximally entangled state. Unlike the optimal states we identified in a single copy of certain irreducible representations, the maximally entangled states have uniform amplitudes on a Dicke state basis.
Note that entanglement between isomorphic unitary representations is also a known resource for unitary tomography. 
Specifically, for covariant estimation of a $d$-dimensional unitary operation $U$ with a uniform prior and query access to $U^{\otimes N}$, it is known that an optimal probe state has support on a direct sum of entangled states of duplicated subrepresentations (of $U(d)$) for each subrepresentation appearing in $(\mathbb{C}^{d})^{\otimes N}$ \cite{PhysRevA.72.042338,PhysRevA.64.050302,PhysRevA.75.022326}.

\section{QFI as a function of $\theta$}
\label{sec:optimality_as_a_function_of_theta}

\begin{figure*}
    \subfloat[$J=3$]{\includegraphics[width =\columnwidth]{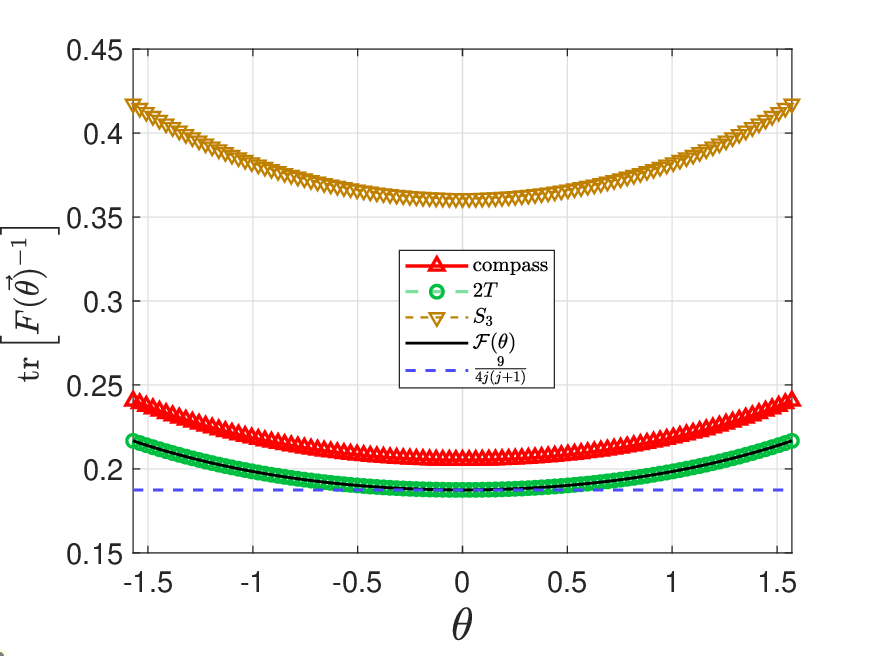} \label{fig:Fig_theta_a}}
    \subfloat[$J=4$]{\includegraphics[width =\columnwidth]{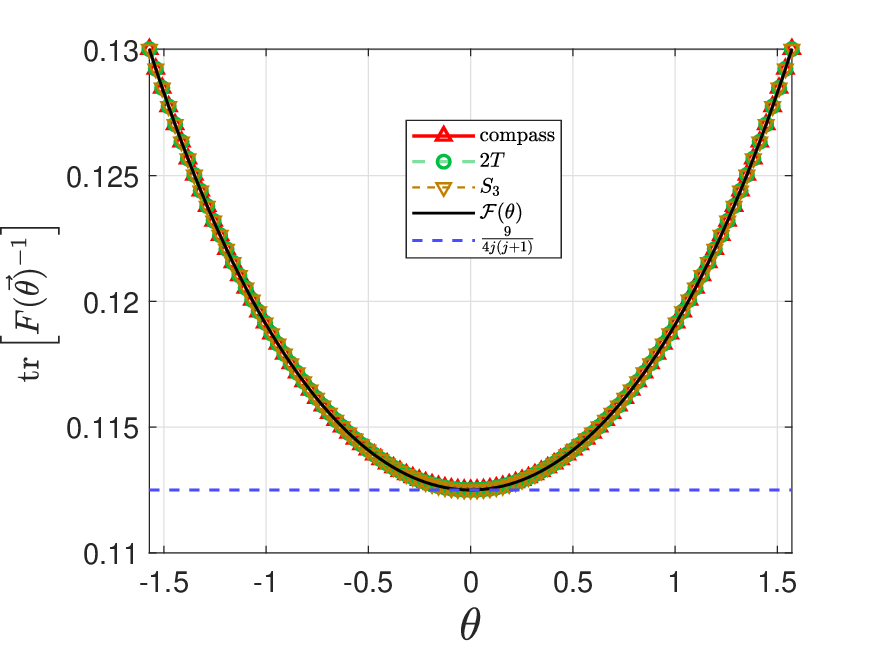} \label{fig:Fig_theta_b}}    
    \caption{The scalar quantum Cram\'{e}r-Rao bound for states considered for $\SU(2)$  estimation as we vary the parameter $\theta$ where $\vec{\theta}=\frac{\theta}{\sqrt{3}}(1,1,1)$.
    In (a), it is seen that the $2T$ state is an optimal state for the $J=3$ irreducible representation, as it satisfies the quantum metrology conditions and therefore obeys \cref{eq:opt_value}.
However, the $S_3$ and compass states do not satisfy the quantum metrology conditions and do not follow the value in \cref{eq:opt_value}.
For the $S_3$ state, even though $J=3$ contains a trivial $S_{3}$ irrep,  the multiplicity is not high enough to satisfy the quantum metrology conditions by fine-tuning (at least two distinct trivial irreps are required, as seen in \cref{eq:find_optimal_state_S3}).
In (b) we consider  $J=4$, where the compass state is a $2T$ state and one can also identify an $S_3$ state that satisfies the quantum metrology conditions.
All one and two-qubit density matrices of these states are given by \cref{eqn:rdms}, so $\mathcal{F}(\theta)$ is given by \cref{eq:opt_value} }
    % \JAG{I think the colormaps need to be made consistent for the Wigner functions, as I don't think the coherent states on the top and the bottom should have primarily negative Wigner functions.
    % Also would be nice to use a diverging colormap, like with white for zero and red/blue for negative/positive (or whatever your preferred convention is).}\Vikas{The dashed line should extend beyond or not touch the circle. I think changing 0 and 1 is a good idea.}}
    \label{fig:comparison_all_angle_1}
\end{figure*}

In this section, we delve into the optimal state for $\SU(2)$ estimation for $\vec{\theta} \neq 0$, which is an important general consideration in multiparameter quantum estimation.
 From the result of~\cite{PhysRevLett.116.030801}, it is known that if an $N$-qubit state has one-qubit and two-qubit reduced states given by
\begin{equation}
    \begin{aligned}
        \rho^{(1)}&=\frac{\mathds{1}}{2},\\
        \rho^{(2)}&= \frac{1}{4} \mathds{1}\otimes \mathds{1}+\frac{1}{12}\sum_{i=1}^3 \sigma_i \otimes \sigma_i,
    \end{aligned}
    \label{eqn:rdms}
\end{equation}
where $\sigma_i$ is the Pauli matrix, then
\begin{equation}
    \text{tr}[F(\vec{\theta})^{-1}] = {3+ {6\over \text{sinc}^{2}\left( {\Vert \vec{\theta}\Vert \over 2}\right) }\over N(N+2)}=: \mathcal{F}(\theta).
    \label{eq:opt_value}
\end{equation}
In \cref{sec:Reduced_density_matrices_for_the_optimal_states}, we prove that the optimal states for the $\SU(2)$ parameter estimation that we have identified at $\vec{\theta}=0$ have the one-body and two-body reduced states as in \cref{eqn:rdms}.

In \cref{fig:comparison_all_angle_1}, we compare  $\mathcal{F}(\theta)$ for the $2T$, $S_3$, and compass states that we analyzed in Section \ref{sec:Multi_parameter_estimation} for $\SU(2)$ estimation. 
For $J=3$, since the $2T$ state is optimal
at $\vec{\theta}=0$ it produces the variance bound given by \cref{eq:opt_value}.
In contrast, the $S_3$ and compass states do not satisfy the quantum metrology conditions and, consequently, give worse lower bounds on variance than \cref{eq:opt_value}.
For the case of $J=4$ in \cref{fig:Fig_theta_b}, the compass state is invariant under $2T$, and a state with $S_3$ symmetry can be found that satisfies the quantum metrology conditions.
As anticipated, for these cases, the $\mathcal{F}(\theta)$ curves satisfy \cref{eq:opt_value}, as the one- and two-qubit density matrices satisfy (\ref{eqn:rdms}) for all three states.

\begin{figure}
    \includegraphics[width =\columnwidth]{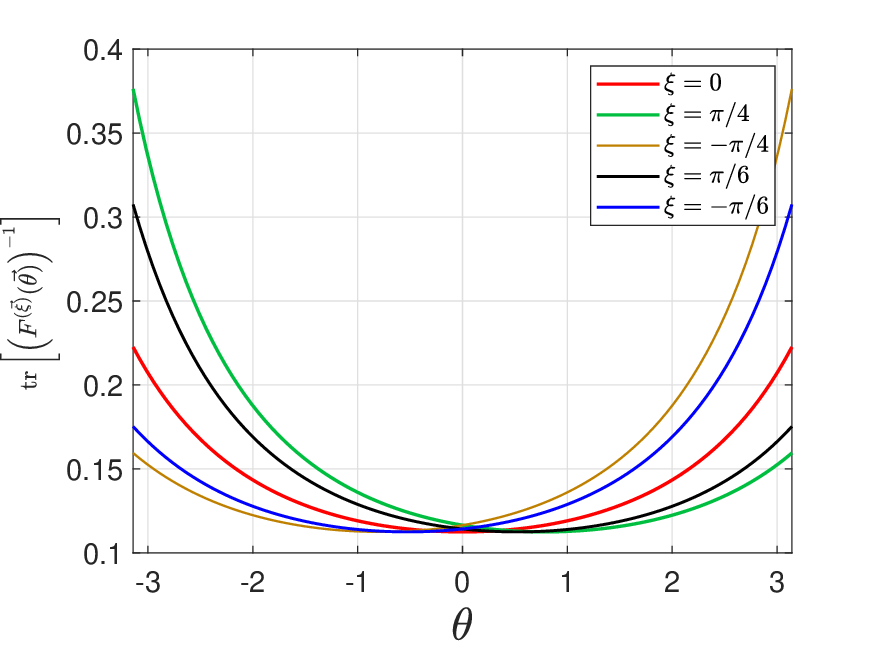}    
    \caption{Scalar quantum Cram\'{e}r-Rao bound for $\SU(2)$ at $\vec{\theta}=\frac{\theta}{\sqrt{3}}(1,1,1)$ for varying $\theta$, and for different values of $\vec{\xi}= \frac{\xi}{\sqrt{3}}(1,1,1)$. 
        We consider the initial state $\ket{\psi}$ in \cref{eqn:fghfgh} to be a $2T$ invariant state in the $J=4$ representation, which is optimal at $\vec{\theta}=0$.
    One sees that the curve can be shifted by changing the value of $\xi$, and thus the optimal state for $\vec{\theta}=0$ is not globally optimal.
One can find a different optimal state for different values of $\vec{\theta}$.
However, one needs to tune the value of $\xi$ to $\vec{\theta}$ to obtain an optimal state, and $\theta$ is the unknown parameter, so a Hamiltonian simulation algorithm with multiple calls to $U(\theta)$ may be required to prepare an optimal state.
    }
    % \JAG{I think the colormaps need to be made consistent for the Wigner functions, as I don't think the coherent states on the top and the bottom should have primarily negative Wigner functions.
    % Also would be nice to use a diverging colormap, like with white for zero and red/blue for negative/positive (or whatever your preferred convention is).}\Vikas{The dashed line should extend beyond or not touch the circle. I think changing 0 and 1 is a good idea.}}
    \label{fig:comparison_all_angle_2}
\end{figure}

For single-parameter estimation, it is known that the GHZ state is globally optimal, i.e., the maximal QFI is obtained for all values of the rotation parameter under consideration.
However, for the case of multiparameter metrology, the parameters are not generally associated with commuting generators and it is not guaranteed that an optimal state at a particular value of the parameters is globally optimal.
In \cref{sec:nonzero-theta} we show that simply transforming a $\vec{\theta}=0$ optimal initial state by an $\SU(2)$ unitary has no effect on the QFI.
For the present case of optimal SU(2) estimation with a symmetric state of $N$ qubits, we now show that having access to at least two interferometer configurations allows us to translate the sensitivity of a probe state for $\SU(2)$ parameter estimation at $\vec{\theta}=0$ to be achieved at arbitrary $\vec{\theta}=\vec{\xi}$. 
Consider $\vec{\xi} \in \mathbb{R}^{3}$ and the probe state 
\begin{equation}
    \ket{\psi_{\vec{\xi}}(\vec{\theta})}:= e^{-i(\vec{\theta}-\vec{\xi})\cdot \vec{J}}\ket{\psi}
    \label{eqn:fghfgh}
\end{equation}
which is to be considered as parametrized by $\vec{\theta}$ only. The vector $\vec{\xi}$ can be interpreted as a known classical magnetic field.
The parametrized state (\ref{eqn:fghfgh}) is not describable by the usual shift model of estimation of $\vec{\theta}$ in which the ``black box'' dynamics $U(\vec{\theta})=e^{-i\vec{\theta}\cdot \vec{J}}$ is applied to an unparametrized probe state, but rather where one has a good guess of the value of the unknown magnetic field and the ability to apply a precise control field to mostly cancel the unknown field.

In the same way as the SLD matrices are derived in (\ref{eqn:aa1}), (\ref{eqn:amatrix}), one finds the SLD matrices for the probe state (\ref{eqn:fghfgh}) to be
\begin{equation}
    L_{\vec{\theta}}^{(j)}=-2iU(\vec{\theta}-\vec{\xi})A^{(j)}(\vec{\theta}-\vec{\xi})\ket{\psi}\bra{\psi}U(\vec{\theta}-\vec{\xi})^{\dagger} + h.c. \, .
\end{equation}
Denoting the QFI matrix with probe state (\ref{eqn:fghfgh}) by $F^{(\vec{\xi})}(\vec{\theta})$, it follows directly from the definition of the QFI matrix that
\begin{equation}
    F^{(\vec{\xi})}(\vec{\theta}) = F^{(0)}(\vec{\theta}-\vec{\xi}).
\end{equation}
Because $\vec{\theta}\cdot \vec{J}$ and $\vec{\xi}\cdot \vec{J}$ do not commute in general, it is not immediately clear how to generate the state (\ref{eqn:fghfgh}) in such a way that $\vec{\theta}$ retains its interpretation as an unknown SU(2) rotation.
However, if we assume that the two interferometer configurations specified by total spin generators $\vec{\theta}\cdot \vec{J}$ and $\vec{\xi}\cdot \vec{J}$ can be utilized, then $\ket{\psi}$ can be formed by alternating applications of $U(\vec{\theta})^{\tau_{1}}$ and $U(\vec{\xi})^{\tau_{2}}$ for sufficiently small $\tau_{1}$, $\tau_{2}$ according to a Hamiltonian simulation protocol such as Trotterization. 
If the allowed interferometer configurations are arbitrary, a variational quantum algorithm with probe state $U(\vec{\theta})R(\vec{\theta},\vec{\xi})\ket{\psi}$ ($R$ varying over $\SU(2)$ is the variational unitary) would also be a viable method for approximately translating the optimal sensitivity of $\ket{\psi}$ to the point $\vec{\theta}$ \cite{Volkoff_2024}.

In \cref{fig:comparison_all_angle_2}, we examine how nonzero $\vec{\xi}$ affects the QFIM. 
We consider vectors of the form $\vec{\xi}=\frac{\xi}{\sqrt{3}}(1,1,1)$ and explore the behavior of the QFI for various values of $\xi$. 
The analysis is shown for  $J=4$, with the state $\ket{\psi}$ in \cref{eqn:fghfgh} taken to be $2T$ invariant, and therefore optimal at $\vec{\theta}=0$.
\section{Measurements that saturate QFI}
\label{sec:measurement_that_saturates_QFI}

Finding a state that optimizes the SLD QFI does not necessarily guarantee an optimal measurement scheme to saturate the QCRB.
Unlike single-parameter estimation, existence of an optimal measurement scheme for multiparameter estimation that saturates the SLD QCRB is not guaranteed.
However, the quantum metrology conditions \cref{eq:condition_theta_0} guarantee that one can construct a measurement that locally saturates the SLD QFI, thereby providing a locally optimal measurement scheme.
% {\color{red}Specifically, the conditions \cref{eq:condition_theta_0} imply that the algebra generated by $\lbrace \mathbb{I}\pm i J_{j}\rbrace_{j=x,y,z}$ is a stabilizer algebra for the code states.}
% \jag{[I don't know what this means.
% When I hear ``stabilizer algebra'' I think of the algebra of Pauli operators that either commute or anticommuting with one another,
% have eigenvalues $\pm1$, and such that you can uniquely define states by their joint eigenvalues for a commuting set of stabilizers.
% What is meant by ``stabilizer'' here?
% And why is the identity operator added to the Lie-algebra elements?]}
Because the operators $J_x,J_y,J_z$ map a code state $\ket{\psi}$ to an orthogonal state, one can construct a measurement to estimate $\vec{\theta}$ as
\begin{equation}
    \{ \ket{\psi}\bra{\psi},P_{1}, P_{2}, P_{3}, Q \}
    \label{eqn:basis}
\end{equation} where
\begin{align}
P_{i}
&:=
{J_{i}\ket{\psi}\bra{\psi}J_{i}\over \Vert J_{i}\ket{\psi}\Vert^{2}}
\end{align}
and $Q$ is a projection that completes the measurement. With
\begin{align}
|\psi(\vec{\theta})\rangle
&:=
\exp(-i\vec{\theta}\cdot \vec{J})|\psi\rangle
\end{align}
one obtains the outcome probabilities
\begin{align}
    p_{0}(\vec{\theta})&:=\big\vert \bra{\psi}\ket{\psi(\vec{\theta})} \big\vert^2,\\
    &= 1-\Vert \vec{\theta}\Vert^2 \frac{J(J+1)}{3} + O(\theta^{4})\\
    p_{i}(\vec{\theta})&:=\langle \psi(\vec{\theta})\vert P_{i} \vert \psi(\vec{\theta})\rangle \\
    &= \theta_i^2{J(J+1)\over 3} +O\left( \theta^{4}\right) \, , \, i=1,2,3 \nonumber\\
    p_{4}(\vec{\theta})&:= \langle \psi(\vec{\theta})\vert Q \vert \psi(\vec{\theta})\rangle \\
    &= O\left( \theta^{4}\right) 
    \,.
\end{align}
From these expressions, one finds that the off-diagonal elements of the classical Fisher information matrix go to zero in the limit $\vec{\theta}\rightarrow 0$ and the diagonal elements are constant at $4J(J+1)/3$ in the same limit. %For example,
% \begin{equation}
% \begin{aligned}
%     F(\vec{\theta})_{xx}&:=\sum_{j=0}^{4} \frac{\left( \partial_{\theta_{x}}p_j \right)^{2}}{p_j},\\
%     &= 4\frac{j(j+1)}{3} \frac{1-(\theta_y^2+\theta_z^2)\frac{j(j+1)}{3}}{1-(\theta_x^2+\theta_y^2+\theta_z^2)\frac{j(j+1)}{3}}+O\left(\theta^{2}\right)  
% \end{aligned}    
% \end{equation}
Thus as $\vec{\theta} \rightarrow 0$, the SLD QCRB is saturated.

\begin{figure}
    \centering
    \includegraphics[width=\columnwidth]{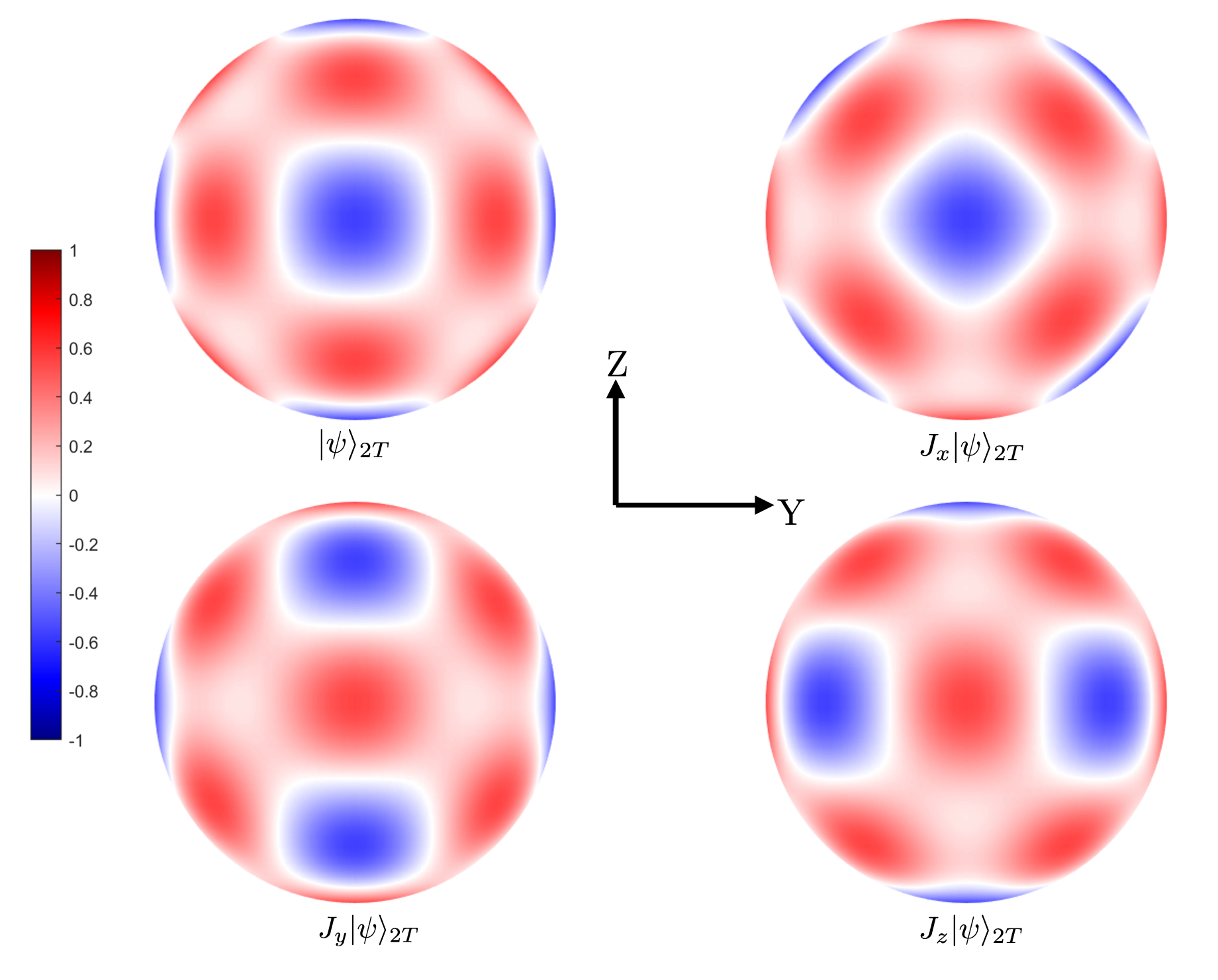}
    \caption{Spin-Wigner function \cite{stratonovich1957distributions,spin_wigner_dowling,spin_wigner_moyal} corresponding to the measurement basis we consider in \cref{eqn:basis} and these are identical to the measurement settings observed in \cite{PRXQuantum.4.020333} for the variational schemes for $\SU(2)$ parameter estimation.}
    \label{fig:Wigner_measurement}
\end{figure}

An alternate approach for obtaining a locally optimal measurement for $\SU(2)$ parameter estimation problem was proposed in \cite{PRXQuantum.4.020333} using the variational schemes for metrology. 
In \cref{fig:Wigner_measurement}, we have plotted the spin-Wigner function \cite{stratonovich1957distributions,spin_wigner_dowling,spin_wigner_moyal} corresponding to the measurement basis we consider in \cref{eqn:basis} and these are identical to the measurement settings observed in \cite{PRXQuantum.4.020333}.
However, unlike the variational approach, in this work using the quantum metrology conditions, one can analytically find an optimal state and measurement.

\begin{figure*}[!ht]
    \subfloat[$J=3$]{\includegraphics[width =\columnwidth]{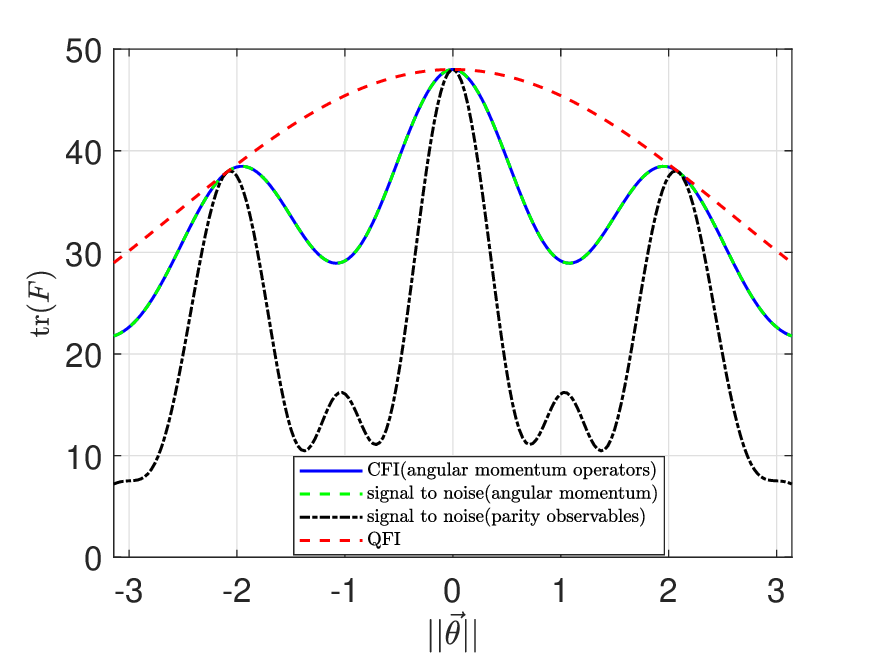} \label{fig:Fig_1_a}}
    \subfloat[$J=4$]{\includegraphics[width =\columnwidth]{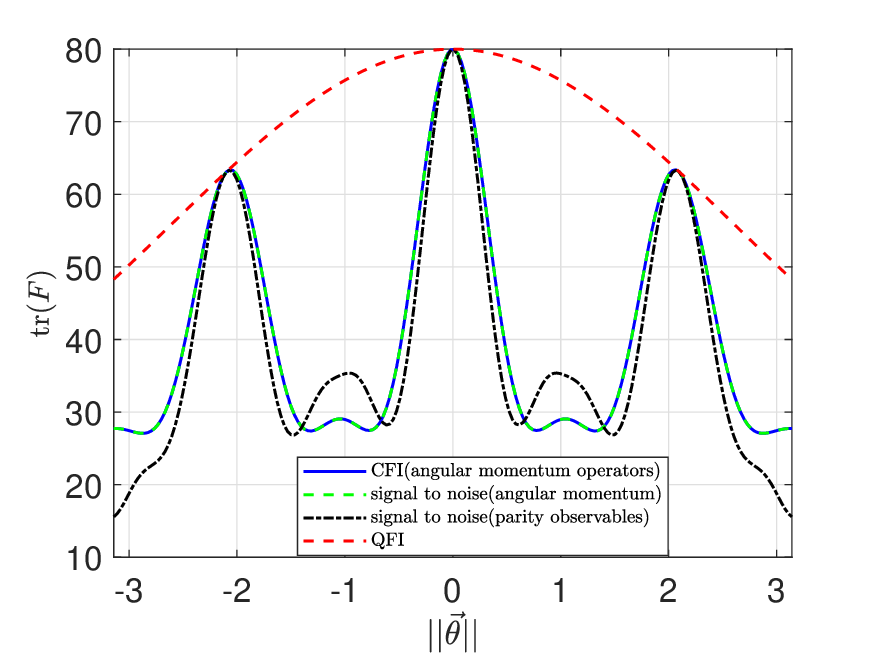} \label{fig:Fig_1_b}}    
    \caption{ Classical Fisher information for the measurement in \cref{eqn:basis} compared to the value of  $\text{tr}M$ given in \cref{eqn:methmom} for an observable list $O$ consisting of the projections obtained from (\ref{eqn:basis}) or the parity observables in \cref{eq:parity_basis}. The probe state is taken to be the $2T$ state  which is optimal at $\vec{\theta}=0$ in the $J=3$ (a) and $J=4$ (b) irreducible  representations.
 We take $\vec{\theta}=\frac{\theta}{\sqrt{3}}(1,1,1)$. For $\theta \to 0$, the reciprocal of the method of moments error of the measurement schemes saturates the QFI.
    }
    % \JAG{I think the colormaps need to be made consistent for the Wigner functions, as I don't think the coherent states on the top and the bottom should have primarily negative Wigner functions.
    % Also would be nice to use a diverging colormap, like with white for zero and red/blue for negative/positive (or whatever your preferred convention is).}\Vikas{The dashed line should extend beyond or not touch the circle. I think changing 0 and 1 is a good idea.}}
    \label{fig:qcrb_saturation}
\end{figure*}

Method of moments estimation provides an alternative way to obtain a matrix that globally majorizes $F(\vec{\theta})^{-1}$ (in the sense of L\"{o}wner ordering). Specifically, the majorizing matrix is defined as a generalized signal-to-noise ratio that depends only on the first moments (i.e., expectation values) and second moments (i.e., the covariance matrix) of a set of operators. To understand this, we briefly review generalized signal-to-noise ratios and their relation to the QFIM. In single-parameter quantum metrology with a probe state $\ket{\psi(\theta)}$, the signal-to-noise ratio of an observable $O=O^{\dagger}$ is defined by
\begin{equation}
{\left(
\frac{d}{d\theta}
\langle O\rangle_{\ket{\psi(\theta)}}
\right)^{2}
\over
\text{Var}_{\ket{\psi(\theta)}}O}
\label{eqn:mommom}
\end{equation}
and provides a practical quantifier of the sensitivity of using the $O$ observable to estimate $\theta$ using the probe state $\ket{\psi}$. It is practical because it does not involve first estimating the probabilities resulting from a measurement of $O$, followed by calculation of the corresponding classical Fisher information. The signal-to-noise ratio (\ref{eqn:mommom}) is an upper bound on $F(\theta)^{-1}$ \cite{Braunstein_caves_metrology,6783980}, although this upper bound is not as tight as that obtained from the classical Fisher information corresponding to a measurement of $O$.

Analogously, in the multiparameter setting we consider estimators of the parameter $\vec{\theta}$ obtained from method-of-moments estimation of a list of observables $O=(O_{1},\ldots,O_{K})$ with $O_{i}=O_{i}^{\dagger}$. When the observables commute,   it is straightforward to show that the following inequality holds \cite{Gessner2020}:
\begin{equation}
 % \left( \sum_{k=1}^{K}M^{(k)}(\vec{\theta}) \right) ^{-1}
 M(\vec{\theta})^{-1}\succeq F(\vec{\theta})^{-1}
\label{eqn:uuu}
\end{equation}
where the matrix elements of $M(\vec{\theta})$ are
\begin{equation}
  M(\vec{\theta})_{i,j}=\partial_{\theta_{i}}\langle O\rangle_{\rho_{\vec{\theta}}}  \left( \text{Cov}_{\rho_{\vec{\theta}}}(O)\right)^{-1}\partial_{\theta_{j}}\langle O\rangle_{\rho_{\vec{\theta}}} 
  \label{eqn:methmom}
\end{equation}
% where the matrix elements of $M^{(k)}(\vec{\theta})$ are
% \begin{equation}
% M^{(k)}(\vec{\theta})_{i,j}:={ \del_{\theta_{i}}\langle O_{k}\rangle_{\rho_{\vec{\theta}}}\del_{\theta_{j}}\langle O_{k}\rangle_{\rho_{\vec{\theta}}}\over \text{Var}_{\rho_{\vec{\theta}}}O_{k}}
% \label{eqn:mom}
% \end{equation}
% for $j=1,2,3$. 
and the covariance is the $K\times K$ matrix
\begin{equation}
    \text{Cov}_{\rho_{\vec{\theta}}}O:=\tr [ \rho_{\vec{\theta}}  \left( (O-\langle O\rangle_{\rho_{\vec{\theta}}})^{T}\circ (O-\langle O\rangle_{\rho_{\vec{\theta}}}) \right) ]
\end{equation}
where $A\circ B:= {1\over 2}[A,B]_{+}$ is the Jordan product.

However, (\ref{eqn:uuu}) also holds when the list $O$ consists of non-commuting observables and when $M$ is nonsingular. To prove (\ref{eqn:uuu}) in this fully general setting, we use the inequality $X^{T}Y(Y^{T}Y)^{-1}Y^{T}X\succeq X^{T}X$ for $X$ a $D\times d$ matrix and $Y$ a full rank $D\times k$ mstrix  \cite{PECARIC1996455}. We define the spectral decompositions $O_{s}=\sum_{\ell_{s}}\lambda_{\ell_{s}}^{(s)}E^{(s)}_{\ell_{s}}$, $s\in [K]$, and define the matrices $X$ and $Y$ by
\begin{align}
    X_{\ell,i}&=\sqrt{p_{\vec{\theta}}(\ell)}\partial_{\theta_{i}}\log p_{\vec{\theta}}(\ell) \nonumber \\
    Y_{\ell,s} &= \sqrt{p_{\vec{\theta}}(\ell)}\left( \lambda_{\ell_{s}}^{(s)}-\langle O_{s}\rangle_{\rho_{\vec{\theta}}}\right).
\end{align}
Note that the probability density $p_{\vec{\theta}}$ is defined on the multi-index $\ell=(\ell_{1},\ldots,\ell_{K})$. The size of the product of the spectra of $O_{1}$, $\ldots$, $O_{K}$ determines $D$.
We define the probability density $p_{\vec{\theta}}(\ell)$ by 
\begin{equation}
p_{\vec{\theta}}(\ell):= {1\over K!}\text{tr}[\sum_{\sigma \in \mathfrak{S}_{K}}E_{\ell_{\sigma(1)}}^{(\sigma(1))}\cdots E_{\ell_{\sigma(K)}}^{(\sigma(K))}\rho_{\theta}]\end{equation} with $\mathfrak{S}_{K}$ the symmetric group on $K$ letters. One can verify that $p_{\vec{\theta}}(\ell)$ is a symmetric function and that, for any subset $L\subset [K]$, the marginal density on indices $\lbrace \ell_{j}\rbrace_{j\in [K]\setminus L}$ is given by 
\begin{equation}
    {1\over (K-\vert L\vert)!}\text{tr}\left[\sum_{\sigma \in \mathfrak{S}_{[K]\setminus L}}\prod_{r\in [K]\setminus L}E_{\ell_{\sigma(r)}}^{(\sigma(r))}\rho_{\vec{\theta}}\right]
\end{equation}
which upon taking $L=[K]$, implies that $p_{\vec{\theta}}(\ell)$ is normalized. 
It is clear that $X^{T}X$ is the classical Fisher information matrix for $p_{\vec{\theta}}(\ell)$ and that $Y^{T}Y$ is $\text{Cov}_{\rho_{\vec{\theta}}}O$. Then one verifies that
\begin{align}
    (Y^{T}X)_{s,i}&=\sum_{\ell}\partial_{\theta_{i}}p_{\vec{\theta}}(\ell)\left( \lambda_{\ell_{s}}^{(s)}-\langle O_{s}\rangle_{\rho_{\vec{\theta}}}\right) \nonumber \\
    &= \sum_{\ell_{s}}\partial_{\theta_{i}}p_{\vec{\theta}}^{(s)}(\ell_{s}) \lambda_{\ell_{s}}^{(s)} \nonumber\\
    &= \partial_{\theta_{i}} \langle O_{s}\rangle_{\rho_{\vec{\theta}}}
\end{align}
where $p_{\vec{\theta}}^{(s)}(\ell_{s}):= \text{tr}E_{\ell_{s}}^{(s)}\rho_{\vec{\theta}}$ is a reduced density of $p_{\vec{\theta}}(\ell)$.

In \cref{fig:qcrb_saturation}, we show the classical Fisher information for the measurement in \cref{eqn:basis}, the value of $\text{tr}M(\vec{\theta})$ given in \cref{eqn:methmom} for an observable list $O$ consisting of the projections (\ref{eqn:basis}), and also the classical Fisher information for the list of parity observables
\begin{equation}
    O=\left( e^{i\pi \left( {N\over 2}\mathbb{I}-J_{x}\right)},e^{i\pi \left( {N\over 2}\mathbb{I}-J_{y}\right)},e^{i\pi \left( {N\over 2}\mathbb{I}-J_{z}\right)} \right).
\label{eq:parity_basis}
\end{equation}
 The probe state is a $2T$ invariant state  and we take $\vec{\theta}=\frac{\theta}{\sqrt{3}}(1,1,1)$.
The numerical results indicate that as $\theta \to 0$, the classical measurement schemes allow for saturation of the QFI.

% \jag{[I confess I still have a hard time following the method-of-moments thing.
% Is this an alternative method to calculate the Fisher information for measuring $(O_1,\ldots,O_K)$ on different copies of the state?
% That is, does it give the same result as taking the convex combination of the Fisher informations for measuring each observable independently?
% Or does this cook up a set of new observables that are measured via these symmetric combinations of the eigenspace projectors?]}

\section{Summary and Discussions}
\label{sec:discussions_and_future_work}

% Quantum metrology stands as a pioneering application within the realm of quantum information science.
% While the majority of research on optimal quantum metrology focuses on single-parameter estimation problems and their implementations,  there is growing interest in multiparameter problems, which can make use of many degrees of freedom of a quantum probe system. 
% For example, given local qubit Hamiltonians $H_{k}$ with tunable coupling parameters $\theta_{k}$, algorithms exist which estimate the $\theta_{k}$ in a parametrized state $e^{-i\sum_{k}\theta_{k}H_{k}}\rho e^{i\sum_{k}\theta_{k}H_{k}}$ to an error $O(1/N)$  with high probability \cite{PhysRevLett.130.200403}. 
% It is therefore of increasing practical importance to identify and generate optimal probe states that allow multiparameter estimation algorithms to be carried out with the highest sensitivity and efficiency.
% \jag{[This first paragraph feels like it belongs more in the introduction as motivation rather than the conclusion.]}

In this work, we combined ideas of multiparameter quantum metrology and quantum error correction to introduce the notion of quantum metrology conditions, which mirror the Knill-Laflamme conditions for quantum error correction.
To make the ideas concrete, we studied the problem of $\SU(2)$ parameter estimation, where one tries to encode parameters into all three directions of rotations.
These conditions allow us to obtain optimal probe states for $\SU(2)$ estimation and other multiparameter estimation problems in which an abelian symmetry relates the generators.

In \cite{PhysRevLett.116.030801}, the compass state was considered as an optimal state for $\SU(2)$ estimation, however, it was found that the compass state was optimal only for the total number of spins ($N$) is such that  $N\equiv 0 \text{ mod }8$.
In this article, we used the quantum metrology conditions to find the optimal state for any value of $N$.
We found that the quantum metrology conditions can be satisfied by states in the trivial irrep of the binary octahedral group ($2T$).
Further, we showed that the compass state lives in the trivial irrep of the $2T$ group for $N\equiv 0 \text{ mod }8$.
We also developed optimal states for $\SU(2)$ estimation using the trivial irrep of the $S_3$ group.
In the context of qubits, an optimal state for $\SU(2)$ parameter estimation, as demonstrated in previous work \cite{PhysRevA.65.012316}, can be identified through entanglement. 
Expanding upon this finding, we establish the general applicability of utilizing entanglement for $\SU(2)$ parameter estimation.

Further, we considered the optimal states as a function of the parameter $\vec{\theta}$ for $\SU(2)$ parameter estimation. 
Unlike the case of the single parameter estimation, the optimal state at $\vec{\theta}=0$, is not globally optimal.
However to find the optimal state for $\vec{\theta}\neq 0$, we need knowledge of the $\vec{\theta}$.
Additionally, the quantum metrology conditions derived in this paper enable the identification of optimal measurement schemes for $\SU(2)$ parameter estimation.
This measurement scheme is particularly effective for small angles of rotation.

One can use the well-known ideas of quantum optimal control \cite{PhysRevA.51.960,PhysRevLett.99.163002,PhysRevA.104.L060401,PRXQuantum.4.040333} or linear combination of unitaries \cite{Berry_LCU,spin_GKP}, to implement the $\SU(2)$ estimation schemes efficiently in the current state of the art experimental settings for quantum metrology \cite{Schine_Kaufman_Bell_state_Martin,Kaubruegger_Zoller_metrology,Bornet_Browaeys_2023_Scalable_sqeezing,Hines_Schleier-Smith_2023_Squeezing_Rydberg_dressing, Eckner_Kaufman_2023_Squeezing_clock}.
We expect that extension of our quantum metrology conditions to more general spin encodings (e.g., symmetric states of $(\mathbb{C}^{d})^{\otimes N}$) and including the effects of decoherence can provide a route toward a unified, symmetry-based framework for identification of optimal noisy probe states for multiparameter quantum estimation.

\begin{acknowledgements}
SO would like to acknowledge fruitful discussions with Ivan Deutsch, Andrew Forbes, Vikas Buchemmavari, and Tyler Thurtel. 
TJV acknowledges useful discussions with Mohan Sarovar.
This work was supported by the Laboratory Directed Research and Development program of Los Alamos National Laboratory under project numbers 20200015ER,  
and the NSF Quantum Leap Challenge Institutes program, Award No. 2016244.

This paper was prepared for informational purposes with contributions from the Global Technology Applied Research center of JPMorgan Chase \& Co.
This paper is not a product of the Research Department of JPMorgan Chase \& Co. or its affiliates. Neither JPMorgan Chase \& Co. nor any of its affiliates makes any explicit or implied representation or warranty, and none of them accept any liability in connection with this position paper, including, without limitation, with respect to the completeness, accuracy, or reliability of the information contained herein and the potential legal, compliance, tax, or accounting effects thereof. This document is not intended as investment research or investment advice, or as a recommendation, offer, or solicitation for the purchase or sale of any security, financial instrument, or financial product or service, or to be used in any way for evaluating the merits of participating in any transaction. 

\end{acknowledgements}
\clearpage

\appendix

\section{Proof of (\ref{eqn:sld})\label{sec:app1}}
The general expression for the symmetric logarithmic derivative of the pure state model $\rho_{\vec{\theta}}=\ket{\psi(\vec{\theta})}\bra{\psi(\vec{\theta})}=U(\vec{\theta})\ket{\psi}\bra{\psi}U(\vec{\theta})^{\dagger}$ is derived in \cite{PhysRevLett.116.030801}, so we just state the result
\begin{equation}
L^{(j)}_{\vec{\theta}}=-2iU(\vec{\theta})A^{(j)}(\vec{\theta})\ket{\psi}\bra{\psi}U(\vec{\theta})^{\dagger} + h.c.
\label{eqn:aa1}
\end{equation}
where
\begin{equation}
A^{(j)}(\vec{\theta}):= \int_{0}^{1}d\alpha e^{i\alpha \vec{\theta}\cdot \vec{J}}J_{j}e^{-i\alpha \vec{\theta}\cdot \vec{J}}.
\label{eqn:amatrix}
\end{equation}
One notes that $A^{(j)}(0)=J_{j}$.
The latter equation is simplified using the adjoint action of SU(2)
\begin{align}
e^{ix\vec{n}\cdot \vec{J}}\vec{m}\cdot \vec{J}e^{-ix\vec{n}\cdot \vec{J}} &= \cos x \,  \vec{m}\cdot \vec{J} + \sin x \, (\vec{m}\times \vec{n})\cdot\vec{J}\nonumber \\
&{} + 2\sin^{2}\left({x\over 2}\right)(\vec{n}\cdot \vec{m})\vec{n}\cdot\vec{J}
\end{align}
to get
\begin{align}
A^{(j)}(\vec{\theta}) &= \int_{0}^{1}d\alpha \, \left[ \cos (\alpha \Vert \vec{\theta}\Vert)J_{j} + 2\sin^{2}\left( {\alpha \Vert \vec{\theta}\Vert \over 2} \right){\theta_{j}\over \Vert \vec{\theta}\Vert^{2}}\vec{\theta}\cdot\vec{J} \right. \nonumber \\
&{} \left. (\vec{e}_{j}\times \vec{\theta}){\sin \alpha \Vert \vec{\theta}\Vert \over \Vert \vec{\theta}\Vert}\cdot \vec{J} \right]\nonumber \\
  &= {\sin \Vert \vec{\theta}\Vert\over \Vert \vec{\theta}\Vert}  J_{j} +\left(1-{\sin\Vert\vec{\theta}\Vert\over \Vert\vec{\theta}\Vert} \right) {\theta_{j}\over \Vert \vec{\theta}\Vert^{2}}\vec{\theta}\cdot \vec{J}\nonumber \\
&{} + 2{\sin^{2}{ \Vert \vec{\theta}\Vert\over 2}\over \Vert \vec{\theta}\Vert^{2}}(\vec{e}_{j}\times \vec{\theta})\cdot \vec{J}.
\label{eqn:aaaa}
\end{align}

To get (\ref{eqn:sld}), note that one can rewrite (\ref{eqn:aa1}) as
\begin{align}
L^{(j)}_{\vec{\theta}}&=2iA^{(j)}(-\vec{\theta})U(\vec{\theta})\ket{\psi}\bra{\psi}U(\vec{\theta})^{\dagger} + h.c. \nonumber \\
&= 2iA^{(j)}(-\vec{\theta})\rho_{\vec{\theta}}+ h.c. \, 
\end{align}
due to the fact that $U(\vec{\theta})A^{(j)}(\vec{\theta})= - A^{(j)}(-\vec{\theta}) U(\vec{\theta})$.

\section{Reduced density matrices for the optimal states}
\label{sec:Reduced_density_matrices_for_the_optimal_states}
In this section, we consider the one and two-qubit reduced density matrices for the optimal states for the $\SU(2)$ parameter estimation at $\vec{\theta}=0$.
We prove that at $\vec{\theta}=0$, the structure of the reduced states given by (\ref{eqn:rdms}) is necessary to achieve the minimal value of trace of the QFI matrix.

\begin{lemma}Let $\ket{\psi}$ be in the spin-$N/2$ representation of $\SU(2)$ given by the symmetric subspace of $N$ qubits. Then $\text{\normalfont{tr}}F(0)^{-1}$ takes the minimal value ${9\over N(N+2)}$ if and only if $\rho^{(1)}:=\text{\normalfont{tr}}_{[N]\setminus \lbrace 1\rbrace}\ket{\psi}\bra{\psi}$ and $\rho^{(2)}:=\text{\normalfont{tr}}_{[N]\setminus \lbrace 1,2\rbrace }\ket{\psi}\bra{\psi}$ are given by (\ref{eqn:rdms}).
    \end{lemma}
    \begin{proof}
       From (\ref{eqn:sldzero}), it follows that $F(0)_{i,j}=4\langle J_{i}-\langle J_{i}\rangle_{\ket{\psi}} \circ  J_{j}-\langle J_{j}\rangle_{\ket{\psi}} \rangle_{\ket{\psi}}$. An optimal state $\ket{\psi}$ can be taken to satisfy $\langle J_{i}\rangle_{\ket{\psi}}=0$ by replacing the rotation generators according to $J_{i}\mapsto J_{i}-\langle J_{i}\rangle_{\ket{\psi}}$. Combined with the fact that $\ket{\psi}$ is symmetric, so that the first part of (\ref{eqn:rdms}) is satisfied. Because $\ket{\psi}$ is symmetric, the two qubit reduced state satisfies $[s,\rho^{(2)}]=0$ where $s$ is in the symmetric group  $\lbrace \mathbb{I}\otimes \mathbb{I} ,\text{SWAP}\rbrace$ on two qubits. Therefore, $\rho^{(2)}$ has the form
       \begin{equation}
           \rho^{(2)}=\begin{pmatrix}
               \lambda_{1}&0&0&0\\
               0&{\lambda_{3}+\lambda_{4}\over 2}&{\lambda_{3}-\lambda_{4}\over 2}&0\\
               0&{\lambda_{3}-\lambda_{4}\over 2}&{\lambda_{3}+\lambda_{4}\over 2}&0\\
               0&0&0&\lambda_{2}
           \end{pmatrix}
       \end{equation}
       in the two qubit computational basis, where $\sum_{l=1}^{4}\lambda_{l}=1$ and $0\le \lambda_{l}\le 1$. Introducing the Pauli matrices $\sigma_{1}$, $\sigma_{2}$, $\sigma_{3}$, note that because $\langle J_{z}\rangle_{\ket{\psi}}=0$, the state $\ket{\psi}$ must satisfy $\sigma_{1}^{\otimes 2}\rho^{(2)}\sigma_{1}^{\otimes 2}=\rho^{(2)}$, which implies that $\lambda_{1}=\lambda_{2}$. It follows that $\rho^{(2)}$ has the general form $\rho^{(2)}=\alpha {\mathbb{I}_{2}\otimes \mathbb{I}_{2}\over 4}+\beta (\sigma_{1}\otimes \sigma_{1} + \sigma_{2}\otimes \sigma_{2}) + \gamma (\sigma_{3}\otimes \sigma_{3})$. The analogous symmetry arguments apply for the other spin directions, which implies that $\beta =\gamma$.  Reparametrizing $\rho^{(2)}$ results in the general form
       \begin{equation}
           \rho^{(2)}= a {\mathbb{I}_{2}\otimes \mathbb{I}_{2}\over 4} + b{\text{SWAP}\over 2}
       \end{equation}
       where $a+b=1$ and $-1\le b\le {1\over 3}$ are required by the unit trace and positivity property, respectively. Explicitly computing the QFIM gives
       \begin{align}
        F(0)_{i,j}&= 2\langle J_{i}J_{j}+J_{j}J_{i}\rangle_{\ket{\psi}}\nonumber \\
           &= {1\over 2}\langle \sum_{\ell,\ell '=1}^{N}\sigma_{i}^{(\ell)}\sigma_{j}^{(\ell ')} +\sigma_{j}^{(\ell)}\sigma_{i}^{(\ell ')} \rangle_{\ket{\psi}}\nonumber \\
           &= N\delta_{i,j}\text{tr}\rho^{(1)}+N(N-1)\text{tr}\left[\rho^{(2)}\sigma_{i}\otimes \sigma_{j}\right] 
       \end{align}      
       so that
       \begin{equation}
           F(0)=  \left( N+Nb(N-1) \right)\mathbb{I}_{3}
       \end{equation}
       where we used the expression for $\rho^{(2)}$ and the fact that $\text{tr}\left[ \text{SWAP}\sigma_{i}\otimes \sigma_{j}\right] = 2\delta_{i,j}$. The minimal value of $\text{tr}F(0)^{-1}$ is obtained for the maximal value $b=1/3$, i.e., $\rho^{(2)}={\mathbb{I}_{2}\otimes \mathbb{I}_{2}+\text{SWAP}\over 6}$ which is the second condition of (\ref{eqn:rdms}), and takes the stated value. 

       The ``if'' part of the lemma is a consequence of  Ref.\cite{PhysRevLett.116.030801}.
    \end{proof}

Next, we prove that the state that satisfies \cref{eq:condition_theta_0} has the one- and two-qubit reduced states in \cref{eqn:rdms}.
From the condition $\langle J_i \rangle =0$, it is straightforward to find that $\langle \sigma_i \rangle=0$ for $i=\{1,2,3\}$, which requires the one-qubit reduced state to be
\begin{equation}
    \rho^{(1)}=\frac{\mathds{1}}{2}.
\end{equation}
To verify that the two-qubit reduced state is given by
\begin{equation}
    \rho^{(2)}= \frac{1}{4} \mathds{1}\otimes \mathds{1}+\frac{1}{12}\sum_{i=1}^3 \sigma_i \otimes \sigma_i \nonumber 
\end{equation}
note that the first quantum metrology condition $\langle J_{i}\rangle=0$ restricts $\rho^{(2)}$ to have the form
\begin{equation}
   \rho^{(2)}= \alpha {\mathbb{I}_{2}\otimes \mathbb{I}_{2}\over 4} + \sum_{i,j}\beta_{i,j}\sigma_{i}\otimes \sigma_{j}.
\end{equation}
But (\ref{eqn:utut}) implies that $\beta_{i,j}={\delta_{i,j}\over 12}$, and the requirement that $\text{tr}\rho^{(2)}=1$ gives $\alpha = {1\over 4}$.
Thus the state that satisfies \cref{eq:condition_theta_0}  has the one- and two-qubit reduced states as in \cref{eqn:rdms}.

\section{SU(2) invariance imposed by quantum metrology conditions}
\label{sec:nonzero-theta}

Here we address the question of whether states satisfying the quantum metrology conditions in \cref{eq:condition_theta_0}, and hence optimal for measuring perturbations away from a fiducial angle $\vec{\theta}=0$, are also optimal for measuring perturbations away from a fiducial angle $\vec{\theta}\neq0$.
A natural question is whether one might translate the optimal QFIM at $\vec{\theta}=0$ to $\vec{\xi}\neq 0$ by applying the inverse unitary $U^\dagger(\vec{\xi})$ to an initial optimal probe state. The parametrized state would then have the form $U(\vec{\theta})U^{\dagger}(\vec{\xi}ß)\ket{\psi}$.
Attempting this, one finds that the failure of the  unitary $U(\vec{\theta})$ to commute with the perturbations we want to measure precludes this naive strategy.
In fact, for states that are optimal at $\vec{\theta}=0$, the QFI remains unchanged when transforming the initial state by an $\SU(2)$ unitary.

To see this, we recall from \cite{PhysRevLett.116.030801} that the QFI for a given state can be expressed as
\begin{align}
    F(\vec{\theta})_{j,k}
    &=
    4\mathrm{Re}\big[\langle A^{(j)}A^{(k)}\rangle-\langle A^{(j)}\rangle\langle A^{(k)}\rangle\big]
    \,,
\end{align}
with $A^{(j)}$ defined as in \cref{eqn:aaaa}.
Since the $A^{(j)}$ do not depend on the initial state, we write the QFI in factorized form, expressing the $A^{(j)}$ through a vector $\vec{A}^{(j)}$ according to
\begin{align}
    A^{(j)}
    &=
    \vec{A}^{(j)}\cdot \vec{J}
    \\
    F(\vec{\theta})_{j,k}
    &=
    4\sum_{m,n=1}^{3}A^{(j)}_m A^{(k)}_n
    \text{Re}\big[
    \langle J_m J_n\rangle
    -\langle J_m\rangle\langle J_n\rangle
    \big]
    \\
    &=
    4\vec{A}^{(j)}\cdot\mathrm{Cov}(\vec{J})\cdot \vec{A}^{(k)}\,.
\end{align}
For an optimal state at $\vec{\theta}=0$ we have from \cref{eq:condition_theta_0} that $\mathrm{Cov}(\vec{J})$ is proportional to the identity, and since this covariance matrix transforms under an $\SU(2)$ transformation of the initial state as
\begin{align}
    \mathrm{Cov}(\vec{J})
    \mapsto
    R^{-1}(\vec{\theta})\mathrm{Cov}(\vec{J})R(\vec{\theta})
    \,,
\end{align}
with $R(\vec{\theta})\in SO(3)$, the QFI remains unchanged under all $\SU(2)$ transformations to the initial state:
\begin{align}
    F(\vec{\theta})_{j,k}
    &=
    4\vec{A}^{(j)}\cdot\vec{A}^{(k)}\langle J_z^2\rangle
    \\
    &=
    4\vec{A}^{(j)}\cdot\vec{A}^{(k)}J(J+1)/3
    \,.
\end{align}

\section{Qudits with $d=4$}
\label{sec:qudits_d_4}
Consider the case of qudits with $d=4$. Define the system of $4\times 4$ matrix units $E_{ij}$ with matrix elements $[E_{ij}]_{st}= \delta_{i,s}\delta_{j,t}$ which satisfy $E_{ij}E_{kl}=\delta_{j,k}E_{il}$. One then notes the following Hermitian basis of the Lie algebra $\mathfrak{u}(4)$ \cite{vilenkin}
% \begin{equation}
%     \begin{aligned}
%         &X_1=\frac{1}{\sqrt{2}}\begin{pmatrix}
%             0&1&0&0\\
%             1&0&0&0\\
%             0&0&0&0\\
%             0&0&0&0
%         \end{pmatrix},\hspace{0.2cm}  
%         X_4=\frac{1}{\sqrt{2}}\begin{pmatrix}
%             0&0&1&0\\
%             0&0&0&0\\
%             1&0&0&0\\
%             0&0&0&0
%         \end{pmatrix},\\
%                  &X_2=\frac{1}{\sqrt{2}}\begin{pmatrix}
%             0&0&0&0\\
%             0&0&0&1\\
%             0&0&0&0\\
%             0&1&0&0
%         \end{pmatrix},\hspace{0.2cm} 
%         X_3=\frac{1}{\sqrt{2}}\begin{pmatrix}
%             0&0&0&0\\
%             0&0&0&0\\
%             0&0&0&1\\
%             0&0&1&0
%         \end{pmatrix}
%     \end{aligned}
%     \label{eq:Gell_mann}
% \end{equation}
\begin{align}
    E_{kk} &{} \; , \; k=1,\ldots,4 \nonumber \\
    Y_{ij}&:= iE_{ij} - iE_{ji} \; , \; i<j \nonumber \\
    X_{ij}&:= E_{ij} + E_{ji} \; , \; i<j
\end{align}
We consider the four-parameter metrology problem defined by the unitary 
\begin{equation}U(\vec{\theta}):= e^{i(\theta_{1}X_{12} + \theta_{2}X_{24} + \theta_{3}X_{34} + \theta_{4}X_{13})}.\end{equation} Note that the generators have the $Z_4$ symmetry 
\begin{align}
    W^{\dagger} X_{12}W &= X_{24} \nonumber \\
    W^{\dagger} X_{24}W &= X_{34} \nonumber \\
    W^{\dagger} X_{34}W &= X_{13} \nonumber \\
    W^{\dagger} X_{13}W &= X_{12} \nonumber \\ 
\end{align}
where the unitary $W$ is,
\begin{widetext}
    \begin{align}
        W&:= \begin{pmatrix}
            0&1&0&0\\
            0&0&0&1\\
            1&0&0&0\\
            0&0&1&0
        \end{pmatrix}  = \exp\left[ i{\pi \over 4} \left( \vphantom{\sum_{0}^{4}} X_{12}+X_{13}+X_{24} + X_{34} - Y_{12} - Y_{24} + Y_{13}+Y_{34}  +\sum_{k=1}^{4}E_{kk}  \right)\right].
\end{align}
\end{widetext}
It is convenient to relabel $X_{0}:= X_{12}$, $X_{1}:=X_{24}$, $X_{2}:=X_{34}$, $X_{3}:=X_{13}$.
The general structure of the QFIM for a probe state invariant under $Z_{4}$ up to a phase can be predicted by noting that the set of QFIM elements $\lbrace F(0)_{i,j} \rbrace_{i,j}$ with $F(0)_{i,j}=2\langle [X_{i}-\langle X_{i}\rangle,X_{j}-\langle X_{j}\rangle]_{+}\rangle$, can be partitioned into three disjoint sets of equal numbers $\lbrace F(0)_{i,i}\rbrace_{i} \sqcup \lbrace F(0)_{i,i\oplus 1}\rbrace_{i} \sqcup \lbrace F(0)_{i,i\oplus 2}\rbrace_{i}$ with $\oplus$ symbolizing modular addition. Therefore, the QFIM is a circulant matrix of the form
\begin{equation}
    F(0)=\begin{pmatrix}
        a & b& c & b \\
        b & a & b & c \\
        c & b & a & b \\
        b & c & b &a 
    \end{pmatrix}.
    \label{eqn:qfim2}
\end{equation}
The eigenvalue of the matrix are given as,
\begin{equation}
\begin{aligned}
\lambda_0&=a+2b+c\\
\lambda_1&=a-c\\
\lambda_2&=a-c\\
\lambda_3&=a-2b+c
\end{aligned}
\end{equation}
The normalized trace of the inverse of the Fisher information matrix at $\vec{\theta}=0$ is calculated as,
\begin{align}
{f}(a,b,c)&:=\mathrm{Tr}(F(0)^{-1}) \nonumber \\
&=\frac{2}{a-c}+\frac{1}{a+2b+c}+\frac{1}{a-2b+c}
\end{align}
To find the minimum of $f(a,b,c)$, first we differentiate the above equation with respect to $b$ which gives as,
\begin{equation}
\frac{\partial {f}(a,b,c)}{\partial b}=\frac{2}{(a-2 b+c)^2}-\frac{2}{(a+2 b+c)^2},
\end{equation}
which on solving we get $b=0$. 
Similarly, we get,
\begin{equation}
\frac{ \partial {f}(a,0,c)}{\partial c}=\frac{2}{(a-c)^2}-\frac{2}{(a+c)^2},
\end{equation}
which on solving we get $c=0$. 
Therefore, in an irreducible representation of $U(4)$ that contains states such that $F(0)$ has the form (\ref{eqn:qfim2}) with $b=c=0$, the optimal state will be obtained by taking $a$ to have the maximal value allowed in that representation, similar to how a multiple of the eigenvalue of the Casimir operator $\sum_{i=1}^{3}J_{i}^{2}$ gives the value $c$ in \cref{eq:condition_theta_0} for $\SU(2)$. By noting that the Casimir invariant $\mathfrak{C}_{2}$ of $\mathfrak{u}(4)$ is a sum of positive operators in the quadratic sector of the universal enveloping algebra, one obtains ${\mathfrak{C}_{2}\over 4}$ as an upper bound for $a$.
The condition for optimality just derived is what we refer to as the quantum metrology condition 
\begin{align}
\langle X_{i} \rangle &= 0\nonumber \\
\langle X_{i}^{2}\rangle &=a \nonumber \\
    \langle X_{i}X_{i\oplus j}\rangle &=0 \; , \; j=1,2 .
        \label{eqn:qmcfull}
\end{align}
Our method for identifying irreducible representations of $U(4)$ which contain probe states that satisfy $F(0)=a\mathbb{I}_{4}$ involves minimally extending the $Z_{4}$ symmetry of the generators to a group $G$ such that invariance under $G$ is sufficient for the conditions (\ref{eqn:qmcfull}) to be satisfied. Noting that the unitary operator
\begin{align}
    Z:= e^{-i{\pi \over 2}(E_{11}-E_{22}-E_{33}-E_{44})}
\end{align}
implements the following symmetry
\begin{align}
    Z^{\dagger}X_{0}Z &= -X_{0} \nonumber \\
    Z^{\dagger}X_{1}Z &= X_{1} \nonumber \\
    Z^{\dagger}X_{2}Z &= X_{2} \nonumber \\
    Z^{\dagger}X_{3}Z &= -X_{3}
\end{align}
we find that a probe state that is invariant under the finite group
\begin{equation}
    G:= \langle W,Z \, \vert \, W^{4}=Z^{4}=(ZW)^{4}=\mathbb{I}\rangle
\end{equation}
is guaranteed to satisfy the quantum metrology condition (\ref{eqn:qmcfull}). Such probe states exist in representations of $U(4)$ that contain trivial irreps of $G$, analogous to how optimal probe states for the three-parameter $\SU(2)$ problem were identified by finding trivial irreps of $A_{4}$ and $S_{3}$ in the totally symmetric representations of $\SU(2)$.

\onecolumngrid
\bibliography{rot.bib}

\end{document}